\documentclass{article}
\usepackage{amsthm}
\usepackage{amsfonts}
\usepackage{xcolor}
\usepackage{amsmath}
\usepackage{amssymb}
\usepackage{bm}
\usepackage{graphicx}
\usepackage{hyperref}
\usepackage{float}
\usepackage{caption}
\usepackage{subcaption}
\usepackage{mathtools}
\usepackage{derivative}
\usepackage[normalem]{ulem}
\usepackage{mathrsfs}
\usepackage{algorithm}
\usepackage{algpseudocode}
\usepackage{pifont}
\usepackage{authblk}
\usepackage{ifthen}
\usepackage{etoolbox} 
\usepackage{imakeidx}
\usepackage{booktabs}
\usepackage[capitalize]{cleveref}
\makeindex

\newtheorem{thm}{Theorem}

\newtheorem{prop}[thm]{Proposition}

\newtheorem{lem}[thm]{Lemma}
\theoremstyle{definition}\newtheorem{rem}[thm]{Remark}
\theoremstyle{definition}\newtheorem{exm}{Example}

\newtheorem{prob}{Problem}

\newcommand{\QQ}{\mathbb{Q}}
\newcommand{\N}{\mathbb{N}}

\newcommand{\RR}{\mathbb{R}}

\newcommand{\Cs}{\mathscr{K}}

\newcommand{\class}{\mathscr{C}}
\newcommand{\Morseclass}{\mathscr{M}}

\newcommand{\ccont}{m}
\newcommand{\llmin}{\lambda}
\newcommand{\llmax}{b}
\newcommand{\spec}{\mathsf{spec}}
\newcommand{\jackreg}{\kappa}
\newcommand{\captureset}{\mathscr{E}}
\newcommand{\bx}{{x}}
\newcommand{\bh}{{h}}
\newcommand{\by}{{y}}

\DeclarePairedDelimiter{\norm}{\lVert}{\rVert}

\def\ud{\textrm{d}} 
\def\Ss{\mathcal{S}} 
  
\def\Pscr{\mathscr{P}}
\def\LL{\mathcal{L}}

\DeclareMathOperator*{\argmin}{argmin}

\DeclareMathOperator*{\pol}{pol}

\DeclareMathOperator*{\crit}{crit_{\Cs}} 
\DeclareMathOperator*{\lmin}{\mathsf{loc\,min}} 

\DeclareMathOperator*{\critset}{\mathsf{crit}}

\newcommand{\ntwo}[2][]{\ifstrempty{#1}{\norm{#2}}{\norm{\, #2\,}_{{#1}}}}

\newcommand{\difd}[1]{\mathrm{d}^{#1}}
\newcommand{\pder}[2]{\frac{\partial #1}{\partial #2}}

\usepackage{xspace}
\newcommand{\poldim}{\mathbb{D}}
\newcommand{\samplesize}{k}
\newcommand{\err}{\mathsf{err}}
\newcommand{\evalpt}{y}
\newcommand{\noise}{\eta}
\newcommand{\bnoise}{{\noise}}
\newcommand{\bbnoise}{\mathsf{p}}
\newcommand{\noiseset}{\mathscr{N}}
\newcommand{\samplepoints}{$\mathsf{SamplePoints}$\xspace}
\newcommand{\posso}{$\mathsf{PSolve}$\xspace}
\newcommand{\precposso}{\mathfrak{e}}
\newcommand{\dlsp}{$\mathsf{DLSP}$\xspace}
\newcommand{\initialize}{$\mathsf{Initialize}$\xspace}
\newcommand{\algogen}{$\mathsf{Minimizers}$\xspace}
\newcommand{\algoreg}{$\mathsf{MinimizersRegular}$\xspace}
\newcommand{\error}{$\mathsf{error}$\xspace}
\newcommand{\complex}{\mathbb{C}}
\newcommand{\coef}{\mathfrak{c}}
\newcommand{\frakpol}{\mathfrak{w}}
\newcommand{\ffield}{\mathbb{K}}

\newcommand{\julia}{\href{https://julialang.org/}{$\mathsf{Julia}$}\xspace}
\newcommand{\msolve}{\href{https://msolve.lip6.fr}{$\mathsf{msolve}$}\xspace}
\newcommand{\globtim}{\href{https://github.com/gescholt/Globtim.jl}{$\mathsf{Globtim}$}\xspace}
\newcommand{\chebfun}{\href{https://www.chebfun.org/}{$\mathsf{chebfun}$}\xspace}
\newcommand{\chebfuntwo}{\href{https://www.chebfun.org/docs/guide/guide12.html}{$\mathsf{chebfun2}$}\xspace}
\newcommand{\chebfunthree}{\href{https://www.chebfun.org/docs/guide/guide18.html}{$\mathsf{chebfun3}$}\xspace}
\date{}

\title{Probabilistic algorithm for computing all local minimizers of Morse functions on a compact domain}
\author[1]{Mohab Safey El Din\thanks{mohab.safey@lip6.fr}}  
\author[2]{Georgy Scholten\thanks{scholten@mpi-cbg.de}}
\author[3]{Emmanuel Trélat\thanks{emmanuel.trelat@sorbonne-universite.fr}}
\affil[1]{Sorbonne Université, CNRS, LIP6, F-75005, Paris, France}
\affil[2]{Max Planck Institute of Molecular Cell Biology and Genetics and
Centre for Systems Biology Dresden, 01307 Dresden, Germany}
\affil[3]{Sorbonne Université, CNRS, Université Paris Cit\'e, Inria, Laboratoire Jacques-Louis Lions (LJLL), F-75005, Paris, France} 

\def\ball{\mathscr{B}}

\def\precvalue{\eta}
\def\precoutput{\varepsilon}
\def\bprecoutput{\mathsf{e}}

\def\pol{w}
\def\Gscr{\mathscr{G}}

\def\chebmeasure{\mu}

\def\revised#1{{{#1}}}
\def\revisedET#1{{{#1}}}
\begin{document}

\maketitle


\abstract{\unboldmath
    Let $\Cs$ be the hyper-cube \([-1, 1]^n\subset \RR^{n}\) and $f: \Cs \to \RR$ be a Morse function.
    We assume that the function $f$ is given by an
    \emph{evaluation program} $\Gamma$ in the \emph{noisy} model, i.e.,  
    the evaluation program $\Gamma$ takes an extra parameter $\precvalue$ as input 
    and returns an  approximation that is $\precvalue$-close to the true value of $f$. 
    In this article, we design an algorithm able to
    compute \emph{all} local minimizers of $f$ on $\Cs$.  
    Our algorithm takes as input $\Gamma$, $\precvalue$, a
    numerical accuracy parameter $\precoutput$ as well as some extra regularity 
    parameters which are made explicit.  Under assumptions of
    probabilistic nature, related to the choice of the evaluation points used
    to feed $\Gamma$, it returns finitely many rational points of
    $\Cs$, such that the set of balls of radius $\precoutput$
    centered at these points contains and separates the set of all local minimizers of $f$. 
    Our method is based on approximation theory, yielding polynomial
    approximants for $f$, combined with computer
    algebra techniques for solving systems of polynomial equations.
    We provide bit complexity estimates for our algorithm when all regularity
    parameters are known. 
    Practical experiments show that our implementation of this algorithm
    in the \julia package \globtim can tackle
    examples that were not reachable until now.
}
\maketitle
\section{Introduction}\label{sec:intro}
\paragraph[short]{Problem statement and regularity assumptions.}
Let $n\in \mathbb{N}\setminus\{0\}$ and let $\Cs$ be the closed hyper-cube $[-1, 1]^n$ in
$\RR^{n}$. Given any sufficiently regular function $f: \Cs \to \RR$ on $\Cs$,
we denote by $\lmin(f)$ the set of \emph{all} local minimizers of $f$ in the interior
of $\Cs$.
Our objective is to design an algorithm that computes $\lmin(f)$, for functions $f$ belonging to
a wide \emph{class of functions} enjoying several regularity properties, that we specify below.

The computational framework we consider is that of an \emph{evaluation model}, where the
function $f$ is only accessible through an evaluation program $\Gamma$.  
Precisely, we assume that $f$ is given by an evaluation program $\Gamma$ which
takes as input a point $\xi$ in $\Cs\cap\QQ^{n}$ and a noise $\noise\in
\QQ_{>0}$, and 
returns $\evalpt\in \QQ$ such that 
$|f(\xi) - \evalpt | \le  \noise$.
When one has an evaluation program that evaluates exactly $f$ on rational
points in  $\Cs$ (i.e., one can take  $\noise = 0$), we 
say that our evaluation model is exact. More details are given further 
on the noisy model. 

Given any $\ccont\in \N$, we denote by $\class^\ccont(\Cs)$ the class of real-valued 
functions on $\Cs$ of which all $i$-th derivatives are continuous for
$0 \leq i\leq m$ (until the boundary). 
When $m\geq 1$, a point $x\in\Cs$ is said to be a \emph{critical point} of $f$
if $\ud f(x)=0$, that is, if the Fr\'echet differential of $f$ at $x$ vanishes; we denote by $x\in\critset(f)$ the set of critical points of $f$.
For $m\geq 2$, we define the class $\Morseclass^m(\Cs)$ of
\emph{Morse functions} on $\Cs$ as the set of functions $f\in\class^\ccont(\Cs)$ 
of which all critical points belong to the interior $\mathring{\Cs}$ of $\Cs$ and are nondegenerate, i.e., 
for every $x\in\critset(f)$, we have $x\in\mathring{\Cs}$ and all eigenvalues of
the Hessian $\ud^2f(x)$ are nonzero.
Note that any $f\in\Morseclass^m(\Cs)$ has finitely many critical points in
$\Cs$ and that $\lmin(f)\subset\critset(f)$; in particular, $\lmin(f)$ is finite.

Given any $\precoutput>0$, we say that \emph{a set $\captureset$ captures $\lmin(f)$ at precision
$\precoutput$} if $\captureset$ is finite and for any $\bx\in \lmin(f)$ there
exists $\by\in \captureset$ such that the closed Euclidean ball $\ball(\by,\precoutput)$ centered at $\by$ of radius
$\precoutput$ contains $\bx$, i.e., $\lmin(f)\subset\ball(\captureset,\precoutput)$.

Given any $f\in \Morseclass^{\ccont}(\Cs)$, we denote by  $\spec(f)$ the set of \emph{all}
eigenvalues of the Hessian of $f$ evaluated at \emph{all} its local minimizers
in the interior of $\Cs$.  Note that $\min\spec(f)>0$.
Given that we will want to approximate quantitatively  $f$ by polynomials enjoying the same property, we should consider classes of Morse functions such that this minimal value is bounded below by some positive constant.
This motivates the following definition.
Given any $\llmin>0$, we define, for any $\ccont\geq 2$,
 \[
     \class^{\ccont}_\llmin(\Cs) = \{f \in \Morseclass^\ccont(\Cs) \mid \min\spec(f)\geq \llmin \} 
.\] 
Given any $\jackreg>0$, we say that $f\in \class^{\ccont}_\llmin(\Cs)$ is $\jackreg$-regular if 
$$
\max_{x\in\Cs} \Vert \ud^j f(x) \Vert  \leq  \jackreg , \quad \forall
1 \le  j \le  m
$$
where $\ud^j f(x)$ is the $j^{\mathrm{th}}$-order Fr\'echet differential of $f$
at $x$, and $\Vert \ud^j f(x) \Vert$ is the operator norm of this multilinear form.
We denote by $\class^{\ccont, \jackreg}_\llmin\left(\Cs  \right) $ the
set of functions in $\class^{\ccont}_\llmin(\Cs)$ that are
$\jackreg$-regular. 

Our objective is to solve the following problem: 

\begin{prob}\label{pbm} 
    Given an evaluation program $\Gamma$ for a function $f\in\class^{\ccont,\jackreg}_\llmin(\Cs) $ 
    and an accuracy parameter $\precoutput>0$, compute a finite set $\captureset$ that
    captures $\lmin(f)$ at precision $\precoutput$. 
\end{prob} 
In this paper, we provide a probabilistic algorithm solving \cref{pbm}. Probabilistic aspects
are inherent to our methodology which is based on approximation theory. 

\medskip

Our initial motivation comes from the motion planning problem in optimal control theory, which
consists of steering a control system from an initial configuration to a final one by
optimizing some criterion under some constraints -- typically, avoid some obstacles (see \cite{laumond,Laumond2016}). 
To determine minimizers, the classical shooting method
boils down to compute zeros of the so-called \emph{shooting
function} (see \cite{trelat_JOTA2012, trelat_SB2024}). 
Usually, this function is not given as a closed formula, but can be
\emph{approximated} at an arbitrary number of points. Hence, when the shooting
function is sufficiently regular, the algorithm described in this
paper allows us to compute all its zeros in the interior of a given
compact domain, and thus, all local extrema of the problem. When one deals with obstacles
to be avoided, this can be particularly interesting to know all possible local minimizers,
because sometimes a local but nonglobal minimizer may enjoy other more interesting properties
such as better stability or robustness.
This motivates the computation of all local minimizers.

\paragraph*{Overall methodology.}
In order to be able to compute a set $\captureset$ that captures $\lmin(f)$
at precision $\precoutput$, we go through the intermediate step of computing a
polynomial approximant of $f$ whose critical points can be
computed exactly. We say that a function $\pol: \Cs \to \RR$ \emph{captures}
$\lmin(f)$ at precision $\precoutput$ if $\critset(w)$ is \emph{finite} 
and $\lmin(\pol)$ captures $\lmin(f)$ at precision $\precoutput$.  
Note that if a set $\Gscr$ captures $\lmin(\pol)$ at
precision $\precoutput/2$ while $\pol$ captures $\lmin(f)$ at precision
$\precoutput/2$, then $\Gscr$
captures $\lmin(f)$ at precision $\precoutput$. 

We exploit the $\jackreg$-regularity assumption to derive suitable
\emph{polynomial approximants} of $f$. These polynomial approximants depend on
some \emph{random} choices of evaluation points, feeding  $\Gamma$ to obtain
approximations of  $f$ up to some noise. This is where probabilistic aspects of
our algorithm lie. This step is called the \emph{polynomial approximation step}. 
To make it efficient, we use \emph{discrete
least-squares approximants}. The quality of these approximants, which is crucial, 
is controlled using the regularity assumptions. 

The relation between the local minimizers of  $f$ in $\mathring{\Cs}$ and
those of the computed polynomial approximant $\pol$ is controlled thanks to a local 
quadratic growth property of  $f$ in neighborhoods of its local minimizers,
which follows from the assumption that $f\in\class^{\ccont}_\llmin(\Cs)$.
It then remains to compute $\critset(w)$.
This is done with algebraic methods ensuring they are all 
computed. This step is called the \emph{algebraic step}. 

\begin{figure}
    \includegraphics[scale=0.5]{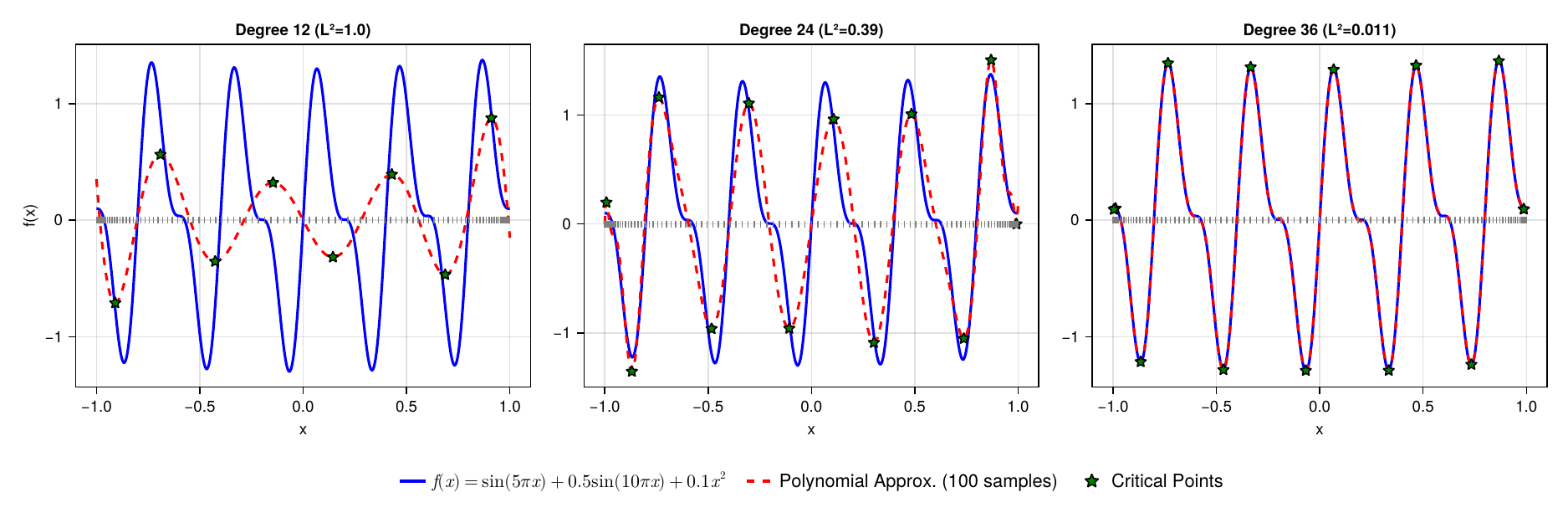}
    \caption{Univariate example}
    \label{fig:univariate}
\end{figure}

\revised{ Figure~\ref{fig:univariate} illustrates this methodology in
    the univariate case. The graph of the function we aim at computing
    its minimizers over $[-1,1]$ is depicted in blue, while, in dotted
    red, are displayed the graph of discrete least squares
    approximants, obtained with randomly chosen sample points. As one
can observe, the minimizers of the approximants converge to those of
the function under study. }

\paragraph*{Main results.}
We will see that for the polynomial approximation step,  
it is \emph{sufficient and advantageous} to work with the $\LL^2$-norm
\begin{equation}\label{def:L2mu}
	\ntwo[\LL^2]{g} = \left(\int_{\Cs}g(x)^2\ud\chebmeasure\right)^{\frac{1}{2}} \qquad\forall g\in \LL^2(\Cs,\chebmeasure),
\end{equation}
where $\chebmeasure$ is the \emph{tensorized Chebyshev measure} $\chebmeasure$ on $\Cs$, defined by
\begin{equation}\label{def:chebyshev_tensorized}
    \chebmeasure(y) = \prod_{i=1}^n \chebmeasure_i(y_i)
\end{equation}
where each $\chebmeasure_i$ is the univariate Chebyshev measure on $[-1,
1]$, of density $\frac{\ud\mu_i}{\ud x_i} = \frac{1}{(1-x_i^2)^{1/2}}$.

Let $\Pscr_{n, d}$ be the $\RR$-vector space of $n$-variate polynomials 
of degree $d$  with real coefficients. It has dimension $\poldim =
\binom{n+d}{n}$. To perform the polynomial approximation step, given a
discrete set $\Ss$ of $\samplesize$ samples (points of $\Cs$), we compute the polynomial 
\begin{equation}\label{eq:discrete_LS}
\pol_{d, \Ss} 
= \argmin_{\pol\in \Pscr_{n, d}}\ \frac{1}{\sqrt{\samplesize}}\bigg(\sum_{s\in \Ss}\big(\pol(s) - \Gamma(s)\big)^2\bigg)^{1/2} ,
\end{equation}
referred to as the \emph{discrete least-squares
polynomial} (DLSP) approximant associated to $\Ss$ at degree  $d$. 
The DLSP $\pol_{d, \Ss}$ is the unique polynomial of degree at most $d$ that
minimizes the sum of the evaluations errors squared over the sample set $\Ss$,
provided that the size of $\Ss$ is large enough, namely that it consists of at least 
$\poldim$
points in general position.  

Setting
\begin{equation}\label{err2}
\err_2(f; d;\chebmeasure) =
\min_{\pol \in \Pscr_{n, d}}\ntwo[\LL^{2}]{\pol -f},
\end{equation}
where $\Vert\ \Vert_{\LL^2}$ is defined by \eqref{def:L2mu} with the measure 
$\chebmeasure$ given by \eqref{def:chebyshev_tensorized},
it is by now well understood in approximation theory to quantify how far is 
$\ntwo[\LL^{2}]{\pol_{d, \Ss} -f}$ from $\err_2(f; d;\chebmeasure)$, in terms of probabilistic considerations,
using randomly sampled sets $\Ss$ of large enough size.
The bulk of our contribution is to leverage such results in order to compute accurate enough
polynomial approximants of $f$.

%

We now provide some details on the noisy model. 
We assume that the noise $\noise (\xi)$ at  $f(\xi)$
for  $\xi\in \Cs$ belongs to some set  $\noiseset\subset \RR$.
We equip $ \Cs\times \noiseset$ with  
the joint measure $\nu$ such that for the above tensorized
Chebyshev measure $\chebmeasure$, we have 
$\chebmeasure \left( B \right) = \nu (B\times \noiseset)$ for any Borel
subset $B$ of  $\Cs$. 
We define the conditional mean of the noise as \[
    \overline{\noise}(\xi) = \mathbb{E}(\noise|\xi)
\] 
and the total expectation of the random variable $ \overline{\noise}^2$ as 
\[
\mathbb{E} \left( \overline{\noise}^2 \right) = \int_{\Cs} \int_{\RR}
\overline{\noise}(\xi)^2 \, \ud \nu (\xi, \noise) = \|\overline{\noise}\|^2 
.\] 
Finally, for $\xi\in \Cs$, we define  $\tilde{\noise}(\xi) = \noise(\xi) -
\overline{\noise}(\xi)$ and 
\begin{equation}\label{noisemax}
    \tilde{\noise}_{\max} = \sup_{\xi\in \Cs} \left( \left |\tilde{\noise}(\xi)\right | \right) 
    \qquad\textrm{and}\qquad
    \overline{\noise}_{\max} = \sup_{\xi\in \Cs} \left( \left |\overline{\noise}(\xi)\right | \right) .
\end{equation}
A classical noisy model is to assume that  
$\tilde{\noise}_{\max}$ and
$\|\overline{\noise}\|$ are both bounded. 

In this paper, our computational assumption is stronger: we assume that the
evaluation program $\Gamma$, that evaluates  $f$, takes as input  $\xi\in
\Cs\cap \QQ^n$
and a second parameter $\bnoise>0$ which can be made arbitrarily small, and
returns  $y\in \QQ$ such that  $|f(\xi)-y| \leq \bnoise$. 

We can now state our main result. We set
\begin{equation}\label{def_beta}
\beta = \frac{\ln(3)}{2\ln(2)}\simeq 0.79 .
\end{equation}

\begin{thm}\label{thm:main:noise}
Let $\ccont\geq \max(3,\beta n+1)$ be an integer, let $\jackreg>0$ and $\llmin>0$ be given positive constants, and let $\precoutput > 0$ be such that $\jackreg\precoutput\leq 3\llmin$.
Let $f\in \class^{\ccont, \jackreg}_\llmin(\Cs)$ be an
arbitrary function, given by an evaluation program in the
noisy model. 
Let $\alpha\in(0, 1)$ be a probability constant.   
Given any $\delta\in(0,1)$, setting
$$
A_1 = 16\left(\frac{1}{\pi^n} +\frac{2}{1-\delta}\right) \frac{C_{n,\ccont}^2\jackreg^2}{(\pi e)^n}
\quad \text{ and } \quad 
A_2 =  \frac{16(1+\delta)}{(\pi e)^n(1-\delta)^2} \left( 8\delta^2(1-\delta)^4+4 \right) ,
$$
where $e$ is the Euler constant and $C_{n,\ccont}$ is the Jackson approximation constant (see Lemma \ref{lemma:jackreg}), if $d\in\N_{\geq 2}$ satisfies
\begin{equation}\label{epsdAi}
    \left( \frac{A_1}{d^{2m-2}} + A_2 \bnoise^2 \right) d^{2\beta n} 
    \leq  \llmin^2\precoutput^4 ,
\end{equation}
then for any set $\Ss$ of $\samplesize$ sample points, randomly (but i.i.d.) chosen according to
the tensorized Chebyshev measure $\chebmeasure$,
with $\samplesize$ large enough so that
\begin{equation}\label{eq:ineq:size}
    \frac{\samplesize}{\ln(\samplesize) + \ln(4/\alpha)} \geq
    \frac{2\poldim^{2\beta}}{\delta^2}, 
 \end{equation}
the DLSP approximant $\pol_{d, \Ss}$ of degree $d$ associated to $\Ss$ of points
captures $\lmin(f)$ at precision $\precoutput$ with probability greater than  $1-\alpha$. 
\end{thm}

A few comments are in order. 

First, this statement implies that, given $m\geq 3$, choosing  $d\geq 2$ large enough and $\bnoise$ small enough so that
$$
d \geq \left( \frac{A_1}{2\llmin^2\precoutput^4} \right)^{\frac{1}{2(m-\beta n -1)}}
\qquad\textrm{and}\qquad
\bnoise \leq \frac{\llmin\precoutput^2}{2 d^{\beta  n}\sqrt{A_2}} ,
$$
and $k$ large enough satisfying \eqref{eq:ineq:size},
we are ensured to obtain a DLSP approximant that captures $\lmin(f)$ at
precision $\precoutput$, with probability greater than
\(1-\alpha\) (under the above assumptions).

Second, the constant $\delta\in(0,1)$ can be used to achieve a trade-off in the estimates \eqref{epsdAi} and \eqref{eq:ineq:size}, depending on the values of $d$, $n$, $\bnoise$ and $\samplesize$.
For instance, when $\delta$ is small then $A_1$ and $A_2$ are lower and thus \eqref{epsdAi} is satisfied for a lower $d$, while \eqref{eq:ineq:size} requires that the number $\samplesize$ of samples be very large. 
At the opposite, when $\delta$ is close to $1$ then the required number $\samplesize$ of samples is lower but the values of $A_1$ and $A_2$ become very large and thus a higher degree $d$ is required to fulfill \eqref{epsdAi}. One can also take $\delta=1/2$ to get estimates.

Third, our statement holds for any function $f\in\class^{\ccont, \jackreg}_\llmin(\Cs)$, 
not only for a single function, i.e., the result stated in \cref{thm:main:noise} is uniform with respect to given parameters $\ccont$, $\jackreg$ and $\llmin$.
From a numerical analysis viewpoint, this is important because 
the result withstands some perturbations within this class of functions.

Fourth, the growth of the degree of the polynomial
approximant is well controlled: when the requested precision becomes
$\precoutput/2$, the required degree is multiplied by $\rho = 4^{1/(m-1-\beta)}$ (independently
on the number $n$ of variables); the accuracy $\bnoise$ on $\Gamma$ also needs to be scaled by 
$\rho^{-\beta n}$ accordingly. 
Here, the dependency is exponential in $n$ which
seems to be unavoidable: the volume of $\Cs$ is exponential in  $n$ too. 
The dependency in  $\alpha$ is explicit in the condition \eqref{eq:ineq:size} on the size $\samplesize$ of
the sample set $\Ss$. Note that when  $\alpha$ is fixed, $\samplesize$ grows polynomially in $d$
and does not need to be one order of magnitude larger than $\mathbb{D}$. 

Last but not least, note that \cref{thm:main:noise} specializes to an exact model 
(corresponding to $\bnoise = 0$). 

\medskip 
Our second set of results is given in \cref{sec:algos}. We derive from
\cref{thm:main:noise} an algorithm, under a technical assumption 
that is easy to satisfy in the noisy model. 
The algorithm takes as input an evaluation program $\Gamma$, evaluating
$f$, a numerical accuracy parameter 
$\precoutput>0$ and a probability constant $\alpha\in(0,1)$. It  
returns a finite set of points with rational coordinates
that captures $\lmin(f)$ at precision $\precoutput$ with probability greater
than $1-\alpha$. 

When the regularity of the function is known, 
the algorithm goes as follows:
\begin{itemize}
    \item choose $d$,  $\bnoise$ and  $\samplesize$ satisfying \eqref{epsdAi} and \eqref{eq:ineq:size}; 
\item pick a sample set $\Ss$ of $\samplesize$ independent and identically distributed 
    (i.i.d.) points of $\Cs$ according to the law of $\chebmeasure$, then call  $\Gamma$ with
    input these sample points and $\bnoise$;
\item compute the least-square polynomial approximant $\pol_{d,\Ss}$ defined by \eqref{eq:discrete_LS};
\item compute all the solutions of the system of polynomial equations 
    $\frac{\partial \pol_{d,
    \Ss}}{\partial x_1} = \cdots = \frac{\partial \pol_{d,
    \Ss}}{\partial x_n} =
    0$ and select and return (rational approximations of) the 
    local minimizers of $\pol_{d,
    \Ss}$ from this solution set. 
\end{itemize}
We analyze the bit complexity of this algorithm, assuming that all regularity
parameters are fixed (as well as the number of variables $n$ since they depend
on $n$), hence focusing on the dependency on the bit size of the output accuracy 
\(\precoutput\) and the
probability parameter \(\alpha\). We prove that, neglecting the calls
to $\Gamma$, the complexity is quasi-linear in the bit size of the probability parameter
$\alpha$ and singly exponential in the bit size of the output accuracy
parameter $\precoutput$
(see \cref{prop:complexity} in \cref{sec:algos}). 

We provide a variant of this algorithm when much less is known on the regularity of
$f$. In particular, we do not assume anymore that 
the regularity parameters (in particular the constants $A_1$ and $A_2$) are known explicitly, but
we only assume that $f$ is a Morse function so that it has finitely many local
minimizers in $\Cs$. Still, this algorithm proceeds by guessing a good enough
polynomial approximant (updating through empirical measures possible values for
$A_1$ and $A_2$) and then computes all critical points of the identified
polynomial approximant. 

This algorithm is implemented within the 
\href{https://julialang.org/}{$\mathsf{Julia}$} package
    \href{https://github.com/gescholt/Globtim.jl}{$\mathsf{Globtim}$} whose
    source code is available at 
    \begin{center}
        \url{https://github.com/gescholt/Globtim.jl}
    \end{center} We
    analyze its practical behavior on a variety of examples in
    \cref{sec:numerical}. We measure how
    good the polynomial approximants are by establishing that, when 
    $f$ is a polynomial of degree $d$, the polynomial approximants have
    degree $d$ too and capture all local minimizers at high accuracy. 
    We also demonstrate on standard benchmarks of global optimization (with
    analytic functions $f$), that our algorithm can solve
    \cref{pbm} on examples that were out of reach of the state of the art until now.

\paragraph*{Related works.}

As we survey the state of the art in global optimization, by which we refer to
computational methods with available software implementation, we want to
emphasize that up to our knowledge, the state-of-the-art numerical global 
optimization methods focus on computing the global minimizer which is not
sufficient to solve \cref{pbm}. 

Up to our knowledge, the closest existing method for solving Problem~\ref{pbm} in
the existing literature that is relevant to our work is the \chebfun
library (see~\cite{chebfun}). 
This \texttt{Matlab} library was originally designed for the construction of univariate
polynomial approximants. It has been extended to the bivariate
case (see~\cite{chebfun}) and more recently to the $3$-dimensional setting
(see~\cite{Dolgov}).
The general approach of \chebfun is to approximate $f$ on a rectangular domain
via a low-rank tensor product of univariate Chebyshev polynomials. 
When $n=2$, this library computes an approximant of the form
\begin{equation}
    \label{eq:chebfun2_decomp}
    f(x, y)\simeq \sum_{i=1}^{r} c_{i} p_i(x)q_i(y),
\end{equation}
where $r$ is the rank and the $c_i$'s are the coefficients of the decomposition.
Although Chebfun seeks to construct an approximant of the lowest rank $r$, the
degrees of the univariate polynomials $p, q$ in this construction are generally
high.
This is intrinsic to the method employed by \chebfun, which constructs bases of
orthogonal polynomials that remain numerically stable up to very high degree
(see~\cite[Chapter~5]{upmc.tref}). 
In the two-dimensional case, \chebfuntwo is able to compute the gradient of this
approximant and to solve for its critical points through a numerical method
relying on resultant computations (see~\cite{Nakatsukasa2015}). 
The high degree of the polynomial system thus obtained would usually put it out
of computational reach for traditional polynomial system solving methods, even
for numerical methods such as Homotopy Continuation (see~\cite{homotopy}).
The numerical solver of \chebfun does not provide guarantees of finding all
solutions, but runs relatively fast.  
Indeed, we encounter these limitations of the purely numerical approach in
multiple examples of Section~\ref{sec:numerical}, see in particular
Examples~\ref{exm:dejong} and \ref{exm:holder_table}, where \chebfuntwo yields
polynomial approximants of a so large degree that this becomes a bottleneck for
solving \cref{pbm}.

In the three-dimensional case, the construction of polynomial approximants based
on tensor products is generalized in \chebfunthree, by means of Tucker
representations (see~\cite{Dolgov}). We report on practical experiments in the
three-dimensional case, showing that the methodology brought by this paper
yields polynomial approximants of smaller degrees than the ones built by
\chebfunthree and still can tackle the
global optimization \cref{pbm} (while \chebfunthree focuses on converging to the
maximum or the minimum of the function). 

Finally, we show in \cref{exm:deuflhard_4d} that our methods scale to the
\(4\)-dimensional case. 

\paragraph*{Acknowledgments.}
The authors thank Albert Cohen, Matthieu Dolbeault, Giovanni
Migliorati and David Papp for insightful discussions and suggestions.  This
research was supported by the grants FA8665-20-1-7029 and FA8655-25-1-7469 of the European Office of
Aerospace Research and Development of the Air Force Office of Scientific
Research (AFOSR). We also thank the reviewers for their suggestions
that helped to improve this paper. 

\section{Proof of \texorpdfstring{\cref{thm:main:noise}}{Theorem \ref*{thm:main:noise}}}\label{sec:proofs:main}
We start by recalling a well known multivariate version of the Jackson approximation theorem, under the $\jackreg$-regularity assumption. Given a function $f$, the best uniform polynomial approximation error of degree $d$ is defined by
\[
    \err_{\infty}(f; d) = \min_{w\in \Pscr_{n, d}}  \| w - f \|_{\LL^{\infty}}.
\]

\begin{lem}\cite{BagbyBosLevenberg_CA2002, Plesniak09, timan2014theory}\label{lemma:jackreg}
There exists $C_{n,\ccont}>0$ (depending only on $n$ and $m$) such that
\[
    \err_{\infty}(f; d) \leq C_{n,\ccont} \frac{\jackreg}{d^{m-1}}  \qquad\forall f\in \class^{\ccont, \jackreg}_\llmin(\Cs).
\] 
\end{lem}

We now establish a criterion to ensure that a polynomial approximant of $f$ captures
 $\lmin(f)$ at a requested precision  $\precoutput$.

\begin{lem}\label{lemma:capture}
    Let $f\in \class^{\ccont, \jackreg}_\llmin(\Cs)$ with $\ccont\geq 3$ and let $\precoutput > 0$ 
    be such that $\jackreg\precoutput\leq 3\llmin$.
    Any $\pol\in \Pscr_{n, d}$ satisfying
    \[
        \|\pol - f\|_{\LL^{\infty}}\leq \frac{\llmin\precoutput^2}{4}
    \] 
    captures $\lmin(f)$ at precision  $\precoutput$. 
\end{lem}

\begin{proof}
    Let $\bx^\star\in\mathring{\Cs}$ be any local minimizer of $f$. 
    For any $\bh \in\RR^n$ such that $\Vert\bh\Vert=\precoutput$, noting that $\ccont\geq 3$,
    it follows from the Taylor formula with integral remainder that
     \begin{equation*}
         \left\vert f(\bx^\star+\bh)  - f(\bx^\star) - \frac{1}{2} \difd{2}f(\bx^\star).(\bh,\bh) \right\vert
         = \bigg\vert \frac{1}{2} \int_0^1 (1-t)^2
         \difd{3}f(\bx^\star+t\bh).(\bh,\bh,\bh)\, \difd{}t \bigg\vert \\
         \leq \frac{\jackreg\precoutput^3}{6} .
     \end{equation*}
    \revisedET{Since $\min\spec(f)\geq \llmin$ by definition of $\class^{\ccont,\jackreg}_\llmin(\Cs)$, we have $\ud^2 f(\bx^\star).(\bh,\bh) \geq \llmin\Vert\bh\Vert^2 = \llmin\precoutput^2$, and thus}
    \[
    f(\bx^\star+\bh) \geq f(\bx^\star) + \llmin\precoutput^2 - \frac{\jackreg\precoutput^3}{6}
    = f(\bx^\star) + \left( \llmin - \frac{\jackreg\precoutput}{6} \right) \precoutput^2 
    \geq f(\bx^\star) + \frac{\llmin}{2} \precoutput^2 
    \] 
    because $\jackreg\precoutput\leq 3\llmin$.
    For any $\pol\in \Pscr_{n, d}$, since $f(\bx^\star+\bh)-\pol(\bx^\star+\bh) \leq \|\pol- f\|_{\LL^\infty}$, we get
    \[
    \pol(\bx^\star+\bh) \geq f(\bx^\star+\bh) - \|\pol- f\|_{\LL^\infty} 
    \geq f(\bx^\star) + \frac{\llmin}{2} \precoutput^2 - \|\pol- f\|_{\LL^\infty} ,
    \] 
    and since $\pol(\bx^\star) - f(\bx^\star) \leq \|\pol- f\|_{\LL^\infty}$, we obtain
    \[
    \pol(\bx^\star+\bh) \geq \pol(\bx^\star) + \frac{\llmin}{2} \precoutput^2 - 2 \|\pol- f\|_{\LL^\infty} 
    .\] 
    Therefore, if $\pol$ satisfies $\|\pol- f\|_{\LL^\infty} \leq \frac{\llmin\precoutput^2}{4}$ 
    then $\pol(\bx^\star+\bh) \geq \pol(\bx^\star)$ for any $\bh \in\RR^n$ such that 
    $\Vert\bh\Vert=\precoutput$, and thus $\pol$ has a local minimizer in the closed ball
    $\ball(\bx^\star,\precoutput)$. As a consequence,     
    $\pol$ captures  $\{\bx^\star\}$ at precision  $\precoutput$.
    Since $\bx^\star$ was an arbitrary local minimizer of $f$ in $\mathring{\Cs}$
    and $f$ has finitely many local minimizers, the conclusion follows.
\end{proof}

We are now in a position to prove \cref{thm:main:noise}. 
Following \cite[Section 2.2]{MIGLIORATI}, for any $\delta\in [0, 1]$ we set $\zeta(\delta) = \delta + (1-\delta)\ln(1-\delta)>0$ and we note that $\zeta$ is an increasing function on $[0,1]$, satisfying $\zeta(0)=0$, $\zeta(1)=1$ and $\frac{\delta^2}{2}\leq\zeta(\delta)\leq\delta^2$
(this function was first introduced in \cite[Theorem 1]{Cohen2013}).
We sample a set $\Ss$ of i.i.d.~points according to the law of the tensorized Chebyshev
measure $\chebmeasure$. 
Since $\samplesize$ satisfies \eqref{eq:ineq:size} by assumption,
using that $\frac{\delta^2}{2}\leq\zeta(\delta)$, we infer that
\begin{equation}\label{cond:size}
    \frac{\samplesize}{\ln(4\samplesize/\alpha)} 
    \geq \frac{2\poldim^{2\beta}}{\delta^2} 
    \geq \frac{\poldim^{2\beta}}{\zeta(\delta)} ,
\end{equation}
where we recall that $2\beta = \frac{\ln(3)}{\ln(2)}$.
Using \eqref{cond:size}, it follows from \cite[Theorem 4 and Lemma 1]{MIGLIORATI} 
(see also \cite[Lemma 3.3]{Chkifa2015})
that the probability (over the random choices of samples) that 
\[
    \| \pol_{d, \Ss} - f\|_{\LL^2}^2 \leq 
      \err_{2} (f; d; \chebmeasure)^2
    + 
    \frac{2}{1-\delta}\err_\infty(f; d)^2
    +
    \frac{8(1+\delta)}{(1-\delta)^{-2}}    
      \frac{\mathbb{D}}{\samplesize}  \ln(4{\samplesize}/\alpha)\, \tilde{\noise}_{\max}^2
    +
    \frac{4(1+\delta)}{(1-\delta)^2}\overline{\noise}_{\max}^2
\]
is larger than $1-\alpha$, where $\err_2(f; d; \chebmeasure)$ is defined by \eqref{err2}, 
and $\tilde{\noise}_{\max}$ and $\overline{\noise}_{\max}$ are defined by \eqref{noisemax}.
Using \eqref{cond:size} and the fact that $\poldim^{1-2\beta}<1$ (because $1-2\beta<0$), 
we have $\frac{\mathbb{D}}{\samplesize}\ln(4{\samplesize}/\alpha) \leq \zeta(\delta) \leq \delta^2$.
Hence, the probability that 
\begin{equation}\label{prwds}
    \| \pol_{d, \Ss} - f\|_{\LL^2}^2 \leq 
      \err_{2} (f; d; \chebmeasure)^2
    + 
    \frac{2}{1-\delta}\err_{\infty}(f; d)^2
    +
    8(1+\delta)(1-\delta)^2 \delta^2 
    \tilde{\noise}_{\max}^2
    +
    \frac{4(1+\delta)}{(1-\delta)^2}\overline{\noise}_{\max}^2
\end{equation}
is larger than $1-\alpha$.


We now compare $\|\pol_{d, \Ss}-f\|_{\LL^\infty}$ with $\|\pol_{d, \Ss}-f\|_{\LL^2}$.
According to \cite[Theorem 7]{ineq_MIGLIORATI}, there holds
\begin{equation}\label{markovineq}
    \|\pol_{d, \Ss}- f\|_{\LL^\infty}^2 \leq \frac{\poldim^{2\beta}}{\pi^n}   
    \|\pol_{d, \Ss} - f\|_{\LL^2}^2 .
\end{equation}
Combining \eqref{prwds} and \eqref{markovineq}, we infer that, if
\begin{equation}\label{cond17:16}
    \frac{\mathbb{D}^{2\beta}}{\pi^n}
    \left (
      \err_{2} (f; d; \chebmeasure)^2
    + 
    \frac{2}{1-\delta}\err_\infty(f; d)^2
    +
    8(1+\delta)(1-\delta)^2 \delta^2 
    \tilde{\noise}_{\max}^2
    +
    \frac{4(1+\delta)}{(1-\delta)^2}\overline{\noise}_{\max}^2
    \right )
    \leq \frac{\llmin^2 \precoutput^4}{16}
\end{equation}
then the probability that $\|\pol_{d, \Ss}-f\|_{\LL^\infty} \leq \frac{\llmin\precoutput^2}{4}$ is larger
than $1-\alpha$ and therefore, by \cref{lemma:capture}, $\pol_{d, \Ss}$ captures  $\lmin(f)$ 
at precision  $\precoutput$ with a probability larger than $1-\alpha$, which is the desired conclusion of the theorem.

We now simplify the sufficient condition \eqref{cond17:16}, by establishing stronger but simpler conditions implying \eqref{cond17:16}.
First, since $\mu$ is the tensorized Chebyshev measure, we have 
$\err_2(f; d; \chebmeasure)^2 \leq \frac{1}{\pi^n}\err_{\infty}(f; d)^2$.
Hence, \eqref{cond17:16} holds true if 
\begin{equation}\label{cond17:19}
    \frac{\mathbb{D}^{2\beta}}{\pi^n}
    \left (
        \left (
    \frac{1}{\pi^n}
    +
        \frac{2}{1-\delta}\right )
    \err_\infty(f; d)^2
    +
    (1+\delta)
    \left (
        8\delta^2(1-\delta)^2
    \tilde{\noise}_{\max}^2
    +
    \frac{4}{(1-\delta)^2}\overline{\noise}_{\max}^2
    \right )
    \right )
    \leq \frac{\llmin^2 \precoutput^4}{16}.
\end{equation}
By Lemma~\ref{lemma:jackreg}, there exists $C_{n, \ccont}>0$ 
such that ${\err_{\infty}\left( f; d \right) \leq C_{n, \ccont}\frac{\jackreg}{d^{\ccont - 1}}}$. Also, by assumption, both $\tilde{\noise}_{\max}$ and  $\overline{\noise}_{\max}$ are bounded above by $\bnoise$. 
Additionally, we note that $\mathbb{D} = \binom{n+d}{n} \leq (ed)^n$ where \(e\)
denotes the Euler constant.\footnote{Indeed,
    $\binom{n+d}{n} \leq \frac{\left( n+d \right)^n}{n!}$ and it is proved by
    induction on \(n\) that \(n!\geq
    \left( \frac{n}{e} \right)^n \) using the fact that \(\left( 1+\frac{1}{n}
    \right)^n \leq  e  \). 
    Hence 
\(\binom{n+d}{n} 
\leq e^n \left(
1+\frac{d}{n} \right)^n \leq (ed)^n \) since \(1+\frac{d}{n} \le d\) for integers
\(n>0\) and \(d\geq 2\).}
Hence, \eqref{cond17:19} holds true if
\[
    \frac{d^{2\beta n}}{(\pi e)^n}
    \left (
        \left (
    \frac{1}{\pi^n}
    +
        \frac{2}{1-\delta}\right )
        \frac{C_{n,\ccont}^2\jackreg^2}{d^{2m-2}}
    +
    (1+\delta)
    \left (
        8\delta^2(1-\delta)^2
    +
    \frac{4}{(1-\delta)^2}
    \right )
    \bnoise^2
    \right )
    \leq \frac{\llmin^2 \precoutput^4}{16} ,
\] 
which is equivalent to \eqref{epsdAi}, when taking
$$
A_1 = 16\left(\frac{1}{\pi^n} +\frac{2}{1-\delta}\right) \frac{C_{n,\ccont}^2\jackreg^2}{(\pi e)^n}
\quad \text{ and } \quad 
A_2 =  \frac{16(1+\delta)}{(\pi e)^n(1-\delta)^2} \left( 8\delta^2(1-\delta)^4+4
\right).
$$
The theorem is proved.
\hfill \qed

\begin{rem}
    Remark that the constant \(A_1\) depends on the Jackson
constant from \cref{lemma:jackreg}, which, itself, depends on the regularity
parameter \(\ccont\). Understanding how \(C_{n,\ccont}\) depends on \(\ccont\)
is a difficult problem -- see e.g. \cite[Chap. 6]{Korneichuk_1991} 
and \cite{Babenko86} for estimates in some special 
cases --  which we do not address in this paper. 
\end{rem}

\section{Algorithms and complexity}\label{sec:algos}

\subsection{Subroutines and algorithm description}

We start by defining some subroutines: 
\begin{itemize}
    \item \dlsp~which takes as input a sequence of points $\Ss = (\xi_1, \ldots,
        \xi_\samplesize)$ in $\QQ^n$, a sequence of rational numbers 
        $\mathscr{Y} = \left( y_1, \ldots, y_k \right) $ 
        an integer $d\in \N$ such that $\samplesize \geq \binom{n+d}{n}$, 
        and returns the polynomial $\pol_{d, \Ss}$ that minimizes 
        \[
            \sum_{i=1}^{\samplesize} \left(  \pol_{d, \Ss}(\xi_i) - y_i
            \right)^2
        .\] 
        Note that computing such a polynomial $\pol_{d, \Ss}$ boils down to
        solving an unconstrained quadratic polynomial optimization problem,
        which is done by solving a linear system of size $\mathbb{D} \times
        \mathbb{D}$ with $\mathbb{D} = \binom{n+d}{n}$.  
    \item \samplepoints~which takes as input an integer $\samplesize$ and
        returns a set of $\samplesize$ i.i.d.~distributed points in  $\Cs$,
        w.r.t. the Chebyshev measure $\chebmeasure$.

        Details on \dlsp~and the rationale on the use of the Chebyshev
        measure $\chebmeasure$ are given in Subsection~\ref{ssec:dlsp}.
    \item \posso~which takes a sequence of polynomials $(g_1 
        \ldots, g_s)$ in $\QQ[x_1, \ldots, x_n]$ and an accuracy para\-meter
        $\precposso$ and returns:
        \begin{itemize}
            \item either \textsf{fail} if and only if 
                the set of common complex solutions to
                the $g_i$'s is not finite;
            \item 
         or a set of points $\bx_1,
        \ldots, \bx_N$ in $\QQ^n$ which approximates at precision $\precposso$ 
        \emph{all} the real solutions to
        the system of equations $g_1 = \cdots = g_s = 0$. 

        This step is done by computing a rational parametrization of the complex
        solutions to the system, i.e., a sequence of polynomials 
        $\left(u, v_1, \ldots, v_n \right)
        $ in $\QQ[t]$ where $t$ is a new variable such that 
        \[
        \left\{ \left( \frac{v_1(\vartheta)}{u'(\vartheta)}, \ldots,
        \frac{v_n(\vartheta)}{u'(\vartheta)}\right)
         \mid 
         u(\vartheta) = 0\right\} 
        \]
        is the set of solutions to $g_1 = \cdots = g_s = 0$. The next step is
        to isolate the real roots of the polynomial $u$ at a
        large enough precision so that plugging in the isolation intervals in
        the rational fractions  $\frac{v_i}{u'}$ yields isolation boxes of size
        bounded by $\precposso$ for the real roots to the input system.

        Details on \posso are given in Subsection~\ref{ssec:posso}.
        \end{itemize}
\end{itemize}
Further, we denote by $\Gamma(\xi, \bnoise)$ the call to the evaluation program
with inputs  $\xi \in \QQ^n$ and the noise accuracy parameter  $\bnoise$. 

Our first algorithm assumes that the quantities $A_1$ and $A_2$  
defined in \cref{thm:main:noise}, which depend on a parameter $\delta\in
(0,1)$, are known. Hence our algorithm actually works
for functions belonging to the class $\class^{\ccont, \jackreg}_\llmin$. 
Since all these regularity quantities are known by assumption, we call this
algorithm \algoreg. Recall that $\beta=\frac{\ln(3)}{2\ln(2)}$. 

\begin{algorithm}[H]
    \caption{\algoreg}
    \label{alg:construct_w_d}
    {
        
    \textbf{Inputs:}

    $\Gamma$: the evaluation program of the objective function $f\in
    \class^{\ccont, \jackreg}_\llmin$,

    $\precoutput$: the output numerical accuracy, 

    $\alpha \in (0,1)$: probability on $\LL^2$-norm of approximant, and 
    $\delta\in (0, 1)$.

    \textbf{Assumption.} $\ccont\geq \max(3,\beta n+1)$. 

    \textbf{Outputs:}

    $\Xi$: a finite set of points which captures $\lmin(f)$ at precision
    $\precoutput$ with probability greater than  $1-\alpha$. 
    }\label{alg:con_deg}

    \begin{algorithmic}[1]
        \State{\label{step:d}
            $d\leftarrow
        \left\lceil  \left( \frac{A_1}{2\llmin^2\precoutput^4} \right)^{\frac{1}{2(m-\beta n -1)}}
    \right\rceil $}
        \Comment{Initialize degree}
        \State{\label{step:bnoise}
            $\bnoise\leftarrow 
\frac{\llmin\precoutput^2}{2 d^{\beta
    n}\sqrt{A_2}}
        $}
        \Comment{Initialize accuracy on $\Gamma$}
        \State{\label{step:k}
            $\samplesize\leftarrow \left\lceil 
                \left(\frac{2\mathbb{D}^{2\beta}}{\delta^2} 
         \left( 1 + \ln(4/\alpha) \right)
         \right)^2
        \right\rceil $ with $\mathbb{D} = \binom{n+d}{n}$}
        \Comment{Initialize sample size}
        \State{\label{step:S}$\Ss\leftarrow$\samplepoints$\left( \samplesize \right) $}
        \State{\label{step:Y}$\mathscr{Y}\leftarrow\{\Gamma\left(\xi ,\bnoise \right)\mid \xi
    \in \Ss\} $}
    \State{\label{step:dlsp}$\pol_{d, \Ss}\leftarrow$\dlsp$\left( \Ss, \mathscr{Y} \right)$ }
    \State{\label{step:posso}\Return
        \posso$\left( \left( \frac{\partial \pol_{d, \Ss}}{\partial x_1},
        \ldots, \frac{\partial \pol_{d, \Ss}}{\partial x_n} \right),
\frac{\precoutput}{2}  \right) $}
    \end{algorithmic}
\end{algorithm}

\begin{rem}
    The dependency of $A_1$ and $A_2$ on $\delta$ is explicitly given in
    \cref{thm:main:noise}. Recall that they depend on the constant  $C_{n, \ccont}$
    defined in \cref{lemma:jackreg}, relating the  $\jackreg$-regularity with
    $\err_\infty(f; d)$. 
\end{rem}

\subsection{Correctness and complexity}

To prove the correctness of Algorithm \algoreg, we use \cref{thm:main:noise} and
the following lemma, which uses the same ingredients of the proof of \cite[Prop.
17]{SaSc18}. This is to control when \posso~returns $\mathsf{fail}$.
Basically, we identify the vector space of polynomials of degree $d$ in
$\complex[x_1, \ldots, x_n]$ with $\complex^{\mathbb{D}}$ with $\mathbb{D} =
\binom{n+d}{n}$ and show that the set of polynomials $\pol\in
\complex^{\mathbb{D}}$ such that the system 
\[
    \frac{\partial \pol}{\partial x_1} = \cdots = \frac{\partial \pol}{\partial
    x_n} = 0
\] 
has infinitely many complex solutions is contained in some Zariski closed subset
of $\complex^{\mathbb{D}}$ for which one can bound its degree with some quantity
depending on $d$ and  $n$.
\begin{lem}\label{lem:genericity}
    For any $d\in \N$, there exists a Zariski closed subset  $\mathscr{Z}\subset
    \complex^{\mathbb{D}}$ (with $\mathbb{D} = \binom{n+d}{n}$), defined by the
    vanishing of a polynomial $\mathscr{C}$
    with coefficients in $\QQ$, of degree bounded by $\left( d+1 \right)^n $, 
    such that for any $\pol\in \complex^{\mathbb{D}}
    \setminus \mathscr{Z}$, the system of equations 
    \[
    \frac{\partial \pol}{\partial x_1} = \cdots = \frac{\partial \pol}{\partial
    x_n} = 0
    \]
    has finitely many solutions in $\complex^n$ and generate a radical ideal.
\end{lem}

\begin{proof}
    We consider indeterminates indexed by the $\mathbb{D}$ monomials of degree
    $\leq d$ in $n$ variables, we denote them by $\coef_{m}$, for $m\in \N^n$
    with $|m|\leq d$. We let $\coef$ the sequence of all these indeterminates
    and $\frakpol = \sum_{|m|\geq 2} \coef_m \bx^m$, using
    the multi-index notation $\bx = x_1^{m_1}\cdots x_n^{m_n}$ with $m = (m_1,
    \ldots, m_n)$. Denote by $\ffield$ the fraction field $\QQ\left(
    \coef \right) $, by $\ffield'$ the rational fraction field generated by
    the $\coef_m$'s with  $|m|\geq 2$ with coefficients in $\QQ$,  $
    \overline{\ffield}$ and $ \overline{\ffield'}$ their algebraic closures, 
    and consider the map
    \[
    \bx \mapsto \left( \frac{\partial \frakpol}{\partial x_1}(\bx), \ldots,
    \frac{\partial \frakpol}{\partial x_n}(\bx) \right) 
    .\] 
    By Sard's theorem \cite[Chap. 2, Sec. 6.2, Thm 2]{Shafa}, the set of
    critical values of this map is contained in a proper Zariski closed subset
    of $\overline{\ffield}^n$. We specifically denote by $\coef_1, \ldots, \coef_n$ the
    indeterminate coefficients which endow the target space of the above map. 
    Hence, the ideal $\mathfrak{I}$ of $\ffield'[x_1, \ldots, x_n, \coef_1,
    \ldots, \coef_n]$ generated by $\frac{\partial
    \frakpol}{\partial x_1}, \ldots,\frac{\partial \frakpol}{\partial x_n}$ and
    the determinant $\mathfrak{D}$ of the Jacobian of this sequence of polynomials contains a
    non-identically zero 
    polynomial in  $\ffield'[\coef_1, \ldots, \coef_n]$. Let $I$ be the ideal of
    $\QQ[x_1, \ldots, x_n, \coef]$ generated by all polynomials in
    $\mathfrak{I}\cap \QQ[x_1, \ldots, x_n, \coef]$. Note that the algebraic
    set associated to this ideal $I$ is the union of some irreducible components of
    the algebraic set defined by  $\frac{\partial
    \frakpol}{\partial x_1}, \ldots,\frac{\partial \frakpol}{\partial x_n}$ and
    $\mathfrak{D}$
    considered in $\QQ[x_1, \ldots, x_n,
    \coef]$ whose projections on the $\coef$-space is not Zariski dense. 
    By Heintz-Bézout's theorem \cite[Theorem 1]{Heintz},
    the sum of the degrees of these irreducible components is bounded above by
    $(d+1)^n$.     We deduce that the sum of the degrees of the Zariski
    closures of the projections of these components on the $\coef$-space is also
    bounded above by  $\left( d+1 \right)^n $. 
    This shows that the radical of the elimination ideal obtained by eliminating
    $x_1, \ldots, x_n$ from $I$ contains non-zero polynomials, one of which
    having degree bounded above by $(d+1)^n$. 

    Now observe that picking a sequence of coefficients $c\in \RR^{\mathbb{D}}
    \setminus \mathscr{Z}$ defines a polynomial $\pol$ such that any complex solution to 
     \[
    \frac{\partial \pol}{\partial x_1}, \ldots,\frac{\partial \pol}{\partial x_n}
    \]
    does not cancel $\mathfrak{D}(c, .)$. In other words, the Jacobian matrix
    associated to the above system is full rank at any of its solutions. By the
    Jacobian criterion \cite[Theorem 16.19]{Eisenbud}, the above system
    generates a radical ideal and has finitely many complex solutions. 
\end{proof}

Further, we say that a polynomial $\pol\in \QQ[x_1, \ldots, x_n]$ is regular if
the system 
\[
\frac{\partial \pol}{\partial x_1} = \cdots = 
\frac{\partial \pol}{\partial x_n} = 0
\] 
generates a radical ideal and has finitely many complex solutions. 

\begin{thm}
Let $\ccont\geq 3$, $\jackreg>0$ and $\llmin>0$ be given, and let $\precoutput > 0$ be an output accuracy parameter such that $\jackreg\precoutput\leq 3\llmin$.
Let $f\in\class^{\ccont, \jackreg}_\llmin$ and $\Gamma$ be an evaluation program evaluating  $f$.
Let $\alpha\in (0, 1)$ be a probability constant and  $\delta\in (0, 1)$. 
    On input $(\Gamma, \precoutput, \alpha, \delta)$, there exists a
    non-identically zero polynomial
    $\mathfrak{C}$ in  $\QQ[s_1, \ldots, s_k, e_1, \ldots, e_k]$ (with $\samplesize$
    defined at Step~\ref{step:k}), of degree bounded by $2d\left( d+1
    \right)^n\binom{n+d}{n} $ (where $d$ is chosen as in Step~\ref{step:d}), 
    such that the following holds. 

    If $\mathfrak{C}\left( \Ss,
    \mathscr{Y} \right) \neq 0 $ (where $\Ss$ and $\mathscr{Y}$ are 
    defined at Steps~\ref{step:S} and~\ref{step:Y} of \algoreg), then 
    Algorithm \algoreg~returns a finite set of points capturing $\lmin(f)$ at
    precision  $\precoutput$ with probability greater than  $1-\alpha$. 
\end{thm}
\begin{proof}
    By \cref{thm:main:noise}, for $d$ and  $\bnoise$ chosen as in
    Steps~\ref{step:d} and~\ref{step:bnoise} of \algoreg, the discrete
    least-squares approximant $\pol_{d, \Ss}$ captures  $\lmin(f)$ at precision
     $\precoutput$, with probability greater than $1-\alpha$, provided that:
      \begin{itemize}
          \item[\textit{(a)}] the size of the sample $\samplesize$ satisfies 
              the inequality \eqref{eq:ineq:size}:
\[
    \frac{\samplesize}{\ln(\samplesize) + \ln(4/\alpha)} \geq
    \frac{2\poldim^{2\beta}}{\delta^2}
    \quad \text{ with } \quad 
    \mathbb{D} =
    \binom{n+d}{n}
    \]
\item[\textit{(b)}] the set $\Ss$ obtained at Step~\ref{step:S} is a set of
    $\samplesize$ i.i.d.~sample points w.r.t. the Chebyshev measure
    $\chebmeasure$,  
\item[\textit{(c)}] the system of polynomial equations 
    \[
        \frac{\partial \pol_{d, \Ss}}{\partial x_1} = 
        \cdots = 
        \frac{\partial \pol_{d, \Ss}}{\partial x_n} = 0
    \]
    has finitely many solutions in $\complex^n$. 
     \end{itemize}
     Condition \textit{(b)} is a direct consequence of the specification of the
     \samplepoints~routine. Hence, we focus on \textit{(a)} and \textit{(c)},
     starting with $\textit{(a)}$.

     We rewrite \eqref{eq:ineq:size} as 
     \[
     \delta^2\samplesize - 2\mathbb{D}^{2\beta} \ln(\samplesize) \ge
     2\mathbb{D}^{2\beta}\ln(4/\alpha)
     .\] 
     Assuming $\samplesize \geq  3$ (so that $\ln(k)\geq 1$), this is a consequence of 
     \[
     \delta^2\frac{\samplesize}{\ln(\samplesize)} \geq 2\mathbb{D}^{2\beta} 
     \left( 1+\ln(4/\alpha) \right) 
     .\] 
     Since $\frac{\samplesize}{\ln(\samplesize)} \geq  \sqrt{\samplesize} $ for $k\geq 2$, the above inequality
     is a consequence of \[
         \samplesize \geq  \left(\frac{2\mathbb{D}^{2\beta}}{\delta^2} 
         \left( 1 + \ln(4/\alpha) \right)
         \right)^2
     \]
     which shows that instantiating $\samplesize$ as in Step~\ref{step:k}
     ensures that \textit{(a)} holds.

     It remains to prove that, under our assumptions, condition \textit{(c)}
     holds. Observe that the discrete least-squares approximant $\pol_{d, \Ss}$
     is obtained by solving a  $\mathbb{D} \times\mathbb{D} $ linear system
     $A \, . \, \coef = b$
     whose coefficients (the entries of $A$ and  $b$) 
     depend on  $\Ss$ and  $\mathscr{Y} = \Gamma(\Ss)$. 
     By Cramer's formul\ae, the solutions  of a $\mathbb{D} \times\mathbb{D} $
     linear system are expressed as rational fractions $\frac{N_i}{D}$ 
     depending on the entries of $A$ and  $b$
     (for $1
     \leq i \leq  \mathbb{D}$ and with $D = \det(A)$) whose coefficients lie in
     the field generated by the coordinates of the points in 
     $\Ss$ and  $\mathscr{Y}$ and whose 
     numerators and denominators have degree at most
     $\mathbb{D}$ in the coefficients of the linear system. To ensure that 
     \textit{(c)} holds, it suffices to ensure that the solutions of this linear
     system (which define the coefficients of $\pol_{d, \Ss}$) do not cancel the
     polynomial  $\mathscr{C}$ defined in \cref{lem:genericity}. 

     Now, remark that the entries of the matrix $A$ and the vector  $b$ are the
     evaluations of polynomials of degree bounded by  $2d$ in  $\Ss$ and
     $\mathscr{Y}$.  All in all, plugging these polynomials in $N_i$ (for  $1
     \leq  i \leq  \mathbb{D}$) and $D = \det(A)$ and next in  $\mathscr{C}$
     defines the polynomial  $\mathfrak{C}$ which has degree bounded by
     $2d(d+1)^n \mathbb{D}$. 

     The last step of the proof is then to establish that $\mathfrak{C}$ is not
     identically zero. To do this, it suffices to prove that for any $d$ and 
     $\samplesize\geq \binom{n+d}{n}$, 
     there exists a set of sample points $\Ss$ and evaluation
     points  $\mathscr{Y}$ such that the associated discrete least-squares
     polynomial approximant  $\pol_{d, \Ss}$ is such that its partial
     derivatives have finitely many complex solutions and generate a radical
     ideal. To do this, we consider the polynomial  
     \[
         \pol = \sum_{i=1}^{n}\left( \frac{x_i^{d+1}}{d+1}-x_i \right) 
     .\] 
     For $\Ss$, we pick randomly $\samplesize$ points in  $\Cs$ and for any
     $\xi\in \Ss$, we choose  $\Gamma(\xi) = w(\xi)$. Since $k\geq
     \binom{n+d}{n}$, it holds that the corresponding polynomial approximant is
      $\pol$ itsself. Now, oberve that the partial derivatives of $\pol$ have
      finitely many complex solutions and generate a radical ideal. 
      This ends the proof.
 \end{proof}

\begin{rem}
    Observe that our technical assumption on the non-vanishing of $\mathfrak{C}$
    is not that restrictive. In the noisy model, it is rather unlikely that
    picking $\Ss$ randomly,  
    $\Gamma(\Ss)$ yields points which cancel  $\mathfrak{C}(\Ss,
    .)$. 
\end{rem}

We establish now the bit complexity of \algoreg. Since it works over the class
of functions $\class_\llmin^{\ccont , \jackreg}$ with source space
$\Cs\subset \RR^n$, we assume that
$\ccont, \jackreg, \llmin, \llmax, \delta$ and $n$ are fixed. {As a
    consequence, the quantities \(A_1\) and \(A_2\) defined in
\cref{thm:main:noise} are also fixed.}
Hence, our complexity
analysis boils down to estimate the degree of the polynomial approximant
computed at Step~\ref{step:dlsp} of \algoreg, and the
bit size of its coefficients, to estimate the bit cost of the call to \posso
at Step~\ref{step:posso} when the input parameter $\precoutput$ varies. We
express the complexity of  \algoreg w.r.t. the bit size precision  of
$\precoutput$, which we denote by  $\bprecoutput$, i.e., we have  $\precoutput =
1 / 2^{\bprecoutput}$ and the bit size $\mathsf{a}$ 
of the probability parameter $\alpha$, i.e., $\alpha  = {1} / {2^\mathsf{a}}$. 

\begin{prop}\label{prop:complexity}
    Assume that the polynomial $\pol_{d, \Ss}$ constructed at
    Step~\ref{step:dlsp} is regular. 

    Then, Algorithm \algoreg~runs in boolean time 
    \[
    \mathrm{O} \left( \left( \ln\left(\mathsf{a}\right) + \mathsf{t}
        +\bprecoutput  \right) 2^{6\bprecoutput n/ (\ccont - \beta n -
         1)}
 \right) 
    \] 
    where \(\mathsf{t}\) is the maximum bit size of the points sampled at
    Step~\ref{step:S}. 
\end{prop}

\begin{proof}
    Applying Step~\ref{step:d} of \algoreg, we have 
     \[
         d\simeq \left( 2^{4\bprecoutput} \frac{A_1}{2\llmin^ 2} \right)^{1/ (2(\ccont - 1 -
         \beta n))} \in \mathrm{O} \left( 2^{2\bprecoutput / (\ccont - \beta n -
         1)} \right) 
    .\] 
    Applying Step~\ref{step:bnoise}, and denoting by $\bbnoise$ the bit size of
    $1 / \bnoise$, i.e.,  $\frac{1}{\bnoise} \simeq 2^{\bbnoise}$, we obtain \[
        2^{\bbnoise} \simeq 2^{2\bprecoutput}2d^{\beta
        n}\frac{\sqrt{A_2}}{\llmin}
                \in 
        \mathrm{O} \left( 2^{2\bprecoutput \left(1 + \frac{\beta n}{\ccont -
        \beta n - 1} \right)}  \right) 
    .\] 
    We now let $\mathsf{a}$ be the bit size of $1 / \alpha$, i.e.,  $\alpha\simeq
    2^{-\mathsf{a}}$. By Step~\ref{step:k}, we deduce that 
    \[
        \samplesize\simeq \left(
            \frac{2\mathbb{D}^{2\beta}}{\delta^2} \left(
        1+\ln(4)+\ln{(1 / \alpha)} \right)
\right)^2 \text{ with } \mathbb{D} = \binom{n+d}{n}   .\]  
   We already observed that \(\mathbb{D} \le (ed)^n\) where \(e\) is the Euler
   constant (see the proof of \cref{thm:main:noise}).
   We deduce that 
   \[
       \samplesize \in \mathrm{O}\left( \mathsf{a} 2^{4\beta n \bprecoutput / (\ccont - \beta n -
         1)}\right) 
   .\] 
    
    We let now $\mathsf{t}$ be the maximum bit size of the sample points in
    $\Ss$. Observe that the monomials of degree $\leq d$ with $n$ variables
    evaluated at such points have bit size at most $\leq d\ \mathsf{t}$ and that
    adding $\samplesize$ of them yields coefficients of bit size in 
    $$\mathrm{O} \left(\ln(k) +
    \mathsf{t} 2^{2\bprecoutput / (\ccont - \beta n -
         1)}
         \right) \subset \mathrm{O}\left(\bprecoutput + \ln(\mathsf{a}) +  \mathsf{t} 2^{2\bprecoutput / (\ccont - \beta n -
         1)}\right)
         \subset \mathrm{O}\left( \ln(\mathsf{a}) +  \mathsf{t} 2^{2\bprecoutput / (\ccont - \beta n -
         1)}\right).$$

    Hence, the coefficients of the $\mathbb{D}\times \mathbb{D}$ linear system
    solved by \dlsp, at Step~\ref{step:dlsp}, have bit size bounded by  
    the maximum of $\bbnoise$ and the above. This lies also in 
    $$\mathrm{O}\left(
\ln(\mathsf{a}) +  \mathsf{t} 2^{2\bprecoutput / (\ccont - \beta n -
         1)}
\right) .$$ 
    Hence, the polynomial $\pol_{d, \Ss}$, computed at
    Step~\ref{step:dlsp}, which has degree 
    $d \in \mathrm{O}\left( 2^{2\bprecoutput / (\ccont - \beta n -
         1)} \right) $
    and has coefficients of bit size which lie in 
    $$\mathrm{O} 
    \left(\mathbb{D} \left( \ln(\mathsf{a}) +  \mathsf{t} 2^{2\bprecoutput /
    (\ccont - \beta n - 1)} \right) \right) \subset 
         \mathrm{O}\left( \left(\ln\left( \mathsf{a} \right) + \mathsf{t}\right) 2^{2(n + 1)\bprecoutput / (\ccont - \beta n - 1)} \right) 
    \quad \text{using} \quad \mathbb{D} \leq  d^n.$$ 
    The complexity of solving the linear system yielding $\pol_{d, \Ss}$ is
    dominated by the one of solving the system of polynomial equations 
     \[
         \frac{\partial \pol_{d, \Ss}}{\partial x_1} = 
         \cdots = 
         \frac{\partial \pol_{d, \Ss}}{\partial x_n} = 0
    .\] 
    Observe that, by our assumptions, this system of polynomial equations
    generates a radical ideal and has finitely many complex solutions. 

    \cite[Corollary 2]{SaSc18} states a bit complexity result on
    computing a rational parametrization of solutions to systems 
    of polynomial equations generating a radical
    ideal with finitely many complex solutions, 
    which is cubic in the Bézout bound attached to the system (here $d^n$) and
    quasi-linear in the maximum bit size of the output coefficients. 
    The output rational parametrization has degree bounded by $\mathcal{D} = d^n$ and its
    coefficients have bit size $\tau'$ where  $\tau'$ lies 
    in  $\mathrm{O} \left( d^n \tau \right) $ where $\tau$
    dominates the maximum bit size of the coefficients of the input polynomials.
    By~\cite[Theorem 47]{MelczerSalvy2021}, isolating the real solutions to such a
    parametrization at bit precision $\bprecoutput$ is done in time which lies in 
    $\mathrm{O} \left( \mathcal{D}^3 + 
    n \left( \mathcal{D}^2\tau' + \mathcal{D} \bprecoutput \right)  \right) $ up
    to logarithmic factors.

    All in all, we deduce that the bit
    complexity of Step~\ref{step:posso} is bounded above by 
\[
    \mathrm{O}\left( \left( \ln\left(\mathsf{a}\right) + \mathsf{t}
        +\bprecoutput  \right) 2^{6\bprecoutput n/ (\ccont - \beta n -
         1)}
 \right) 
.\] 
\end{proof}

\begin{rem}
    Our complexity analysis shows that, under the assumption that
    all regularity parameters are known, the complexity of our algorithm is
    singly exponential in the bit size $\bprecoutput$ of the 
    output precision parameter $\precoutput$
    and linear in the logarithm of the bit size $\mathsf{a}$ of the probability parameter
    $\alpha$. This establishes that, fixing the output precision parameter,
    increasing the probability to capture $\lmin(f)$ is not computationally
    expensive. 

    Besides, recall that $\precoutput = 1 / 2^{\bprecoutput}$. Hence, our complexity
    statement  shows that \algoreg~runs in time which is polynomial in $1 /
    \precoutput$, again assuming that the regularity parameters on the input are
    known and fixed.

    Analyzing the complexity of \algoreg~when e.g.  $n$ varies is a task which
    goes way beyond the goals of this paper.  Doing so requires to know how the
    constants $A_1$ and $A_2$ defined in \cref{thm:main:noise} vary w.r.t.  $n$. 
    Here, there is no miracle to expect: at least when the input function $f$ is
    a polynomial, one cannot hope better than a complexity which is exponential
    in $n$ (see e.g. \cite{BGHS14}). 
\end{rem}

\subsection{Adaptive algorithm}

Algorithm \algoreg~is based on the assumption that the regularity parameters
attached to the class of functions 
$\class_\llmin^{\ccont, \jackreg}$, defined as the constants in
\cref{thm:main:noise} are known. This is a strong assumption: in many cases,
such constants are not known from the end-user. 

Now, we design an adaptive algorithm which is based on similar
paradigms than the ones used in \algoreg, but trying to guess the regularity
constants $A_1$ and $A_2$ using extra computations to measure how far the
polynomial approximant, which is considered, is to the one \algoreg~would have
computed if $A_1$ and $A_2$ were known. 

We denote by \initialize a subroutine which takes as input 
an evaluation program $\Gamma$, $n, \llmin, \alpha,
\precoutput$, as well as values $A_1, A_2$ 
and applies Steps~\ref{step:d} to
\ref{step:dlsp} using the input values of $A_1, A_2$ instead 
of the quantities depending on the
parameter \(\delta\in (0,1)\), defined in \cref{thm:main:noise}. 
\revisedET{For the sample-size bound, one can take $\delta=1/2$ (see the discussion after \cref{thm:main:noise}) in Step~\ref{istep:k} below.} 

Additionally,
we choose $\ccont = \max(3, \beta n + 1) $, since the value given to \(d\) at
Step~\ref{i:step:d} depends on \(\ccont\). This choice allows us to be in
position to apply \cref{thm:main:noise} as soon as \(A_1\) and \(A_2\) are large
enough, for the largest set of functions, i.e. the smallest regularity
requirements. 
The routine \initialize then outputs the
polynomial approximant $\pol_{d, \Ss}$ of Step~\ref{step:dlsp} as well as the
size $\samplesize$ of the sample from
Step~\ref{step:k} and $\bnoise$ the admissible noise of Step~\ref{step:bnoise}
on the evaluations provided by  $\Gamma$. 

\begin{algorithm}[!htbp]
    \caption{$\textsf{Initialize}$}
    \label{alg:initialize}
    {
        
    \textbf{Inputs:}

    $\Gamma$: the evaluation program of the objective function $f\in
    \class^{\ccont, \jackreg}_\llmin$,

    $n$ the dimension of the ambient space and 
     $\precoutput$, the accuracy parameter 

    $\alpha \in (0,1)$: probability on $\LL^2$-norm of approximant, and
     constants $\llmin, A_1$ and $A_2$

    \textbf{Assumption.} $\ccont\geq \max(3,\beta n+1)$.

    \textbf{Outputs:}

Polynomial approximant $\pol_{d, \Ss}$, 
Sample size $\samplesize$, 
Noise  $\bnoise$
    }\label{alg:icon_deg}

    \begin{algorithmic}[1]
        \State{\label{i:step:d}
            $d\leftarrow
        \left\lceil  \left( \frac{A_1}{2\llmin^2\precoutput^4} \right)^{\frac{1}{2(m-\beta n -1)}}
    \right\rceil $}
        \Comment{Initialize degree}
        \State{\label{istep:bnoise}
            $\bnoise\leftarrow 
\frac{\llmin\precoutput^2}{2 d^{\beta
    n}\sqrt{A_2}}
        $}
        \Comment{Initialize accuracy on $\Gamma$}
        \State{\label{istep:k}
            $\samplesize\leftarrow \left\lceil 
                \left(\frac{2\mathbb{D}^{2\beta}}{\delta^2} 
         \left( 1 + \ln(4/\alpha) \right)
         \right)^2
        \right\rceil $ with $\mathbb{D} = \binom{n+d}{n}$}
        \Comment{Initialize sample size}
        \State{\label{istep:S}$\Ss\leftarrow$\samplepoints$\left( \samplesize \right) $}
        \State{\label{istep:Y}$\mathscr{Y}\leftarrow\{\Gamma\left(\xi ,\bnoise \right)\mid \xi
    \in \Ss\} $}
    \State{\label{istep:dlsp}$\pol_{d, \Ss}\leftarrow$\dlsp$\left( \Ss, \mathscr{Y} \right)$ }
    \State{\Return $\pol_{d, \Ss}$,  $\samplesize$ and
    $\bnoise$}
    \end{algorithmic}
\end{algorithm}
We name our adaptive algorithm \algogen. It takes as input an evaluation
program $\Gamma$, evaluating our objective function $f$ over  $\Cs$ (we assume
that  $n$ is known and that $f$ is a Morse function),  
the accuracy parameter $\precoutput$ and the
probability parameter  $\alpha$, as well as a tolerance parameter which is
used as a tolerance parameter on the error in norm $\LL^2$ achieved by the 
polynomial approximants we construct.

\begin{algorithm}[!htbp]
    \caption{\algogen}
    \label{alg:fgen}
    {
        
    \textbf{Inputs:}

    $\Gamma$: the evaluation program of the objective function $f$ which is
    assumed to be Morse,

    $\precoutput$: the output numerical accuracy, 

    $\alpha \in (0,1)$: probability on $\LL^2$-norm of approximant, and 
    $t\in (0, 1)$: a tolerance parameter.

    \textbf{Outputs:}

    $\Xi$: a finite set of points which captures $\lmin(f)$ at precision
    $\precoutput$ with probability greater than  $1-\alpha$, assuming the
    heuristic criterion to guess the reliability of our polynomial approximant
    is correct. 
    }\label{alg:fcond2}

    \begin{algorithmic}[1]
        \State{$\llmin \leftarrow 1 / 2^{16}$
        and $b \leftarrow \mathsf{true}$}
        \State{$A_1= 2$ and $A_2 = 2$}
        \While {b}
        \State{$\pol, \samplesize, \bnoise \leftarrow$\initialize 
            $\left (\Gamma, n, \alpha,
        \precoutput, \llmin, A_1, A_2\right )$}
        \State{$\Ss\leftarrow$\samplepoints$\left( \samplesize \right) $}
        \State{$\vartheta\leftarrow$\error$(\pol, \Gamma,
        \Ss)$}\label{alg:gen:error} 
        \State{$\Ss'\leftarrow$\samplepoints$\left(2 \samplesize \right) $ and 
        $\mathscr{Y}\leftarrow\{\Gamma\left(\xi ,\bnoise / 2 \right)\mid \xi
    \in \Ss'\} $}\label{alg:gen:newsamples}
    \State{$\pol_{d+1, \Ss'}\leftarrow$\dlsp$\left( \Ss',
    \mathscr{Y} \right)$ and $\vartheta'\leftarrow$\error$(\pol_{d+1,
\Ss'}, \Gamma, \Ss')$}\label{alg:gen:newpol}
    \If{$\vartheta < t$ and $\vartheta'< t$}
        \State{ $b \leftarrow \mathsf{false}$}
    \Else
    \State{$A_1\leftarrow 2 A_1, \, A_2 \leftarrow 2 A_2, \llmin \leftarrow
    \llmin / 2$}
    \EndIf
        \EndWhile
    \State{\label{step:posso2}
        \Return \posso$\left( \left( \frac{\partial \pol}{\partial x_1},
        \ldots, \frac{\partial \pol}{\partial x_n} \right),
\frac{\precoutput}{2}  \right) $}
    \end{algorithmic}
\end{algorithm}

It starts by calling \initialize with initial values for  $\llmin, A_1, A_2$
to produce a first polynomial approximant. 
Since we do not assume that we have a full knowledge of the regularity 
of $f$, we need a criterion to decide whether the value $\llmin$ we chose  is
small enough and the values we chose for $A_1$ and $A_2$ are large enough. 

To do so, we first compute the error $\vartheta$ in norm $\LL^2$ over $\Ss$ from the
evaluation of  $\pol_{d, \Ss}$ and  $\Gamma$ at the points of $\Ss$. This
is done through a routine which we name \error. 
Next, we pick a set of new sample points in $\Cs$ of size $2\samplesize$ and
produce a polynomial approximant of degree $d+1$ and again compute the
corresponding error $\vartheta'$ in norm $\LL^2$. 
If both  $\vartheta$ and  $\vartheta'$
are below $\precoutput$, we proceed with the polynomial system solving step. 

Otherwise, we increase  $A_1$ and $A_2$ and decrease $\llmin$.

We report in \cref{sec:numerical} on experiments implementing this strategy when
the tolerance parameter is set to $\precoutput$.  
Note that, even with such a choice, the above criterion remains a heuristic type
but is expected to be reliable on regular enough functions. 
However, if this heuristic criterion leads to pick values for $A_1, A_2$ and
$\llmin$ which are good enough, \cref{thm:main:noise} applies.

\revised{
\begin{rem}
    We emphasize that the behaviour of our adaptive algorithm very
    much depends on the heuristic choice to estimate regularity
    constants. For sure, when the constant $\llmin$ is not small
    enough and/or when $A_1$ and $A_2$ are not large enough, the
    algorithm may miss local minimizers. \\
    The bulk of the design of Algorithm \ref{alg:fgen} lies in the
    countermeasures it proposes to \emph{detect} at Steps 
    \ref{alg:gen:error}, \ref{alg:gen:newsamples} and
    \ref{alg:gen:newpol}, which basically double the size of samples
    and compute a polynomial approximant of one more degree, checking
    that the corresponding error in norm $\mathcal{L}^2$ is close to
    the one under consideration.  \\
    Such a criterion is faked when, for the function under study,
    there are two consecutive discrete least square polynomial approximants,  
    chosen with different samples, with $\epsilon$-close relative
    errors in norm $\mathcal{L}^2$, but whose minimizers are still far
    from the ones of the function under study. This would occur when,
    almost everywhere on the hypercube, the graphs of the function and
    the polynomial approximants are rather close, except at rather
    small neighbourhoods of some of the true minimizers. Hence, the function
    would behave like a Dirac.\\
\end{rem}}

\revisedET{ 
\begin{rem}[Practical validation of the adaptive output.]\label{rem:adaptive_validation}
Since Algorithm~\ref{alg:fgen} relies on a heuristic stopping criterion, it is good practice to perform an \emph{a posteriori} validation step. Typical safeguards include:
\begin{itemize}
\item rerunning the algorithm with independent sample sets (different random seeds) and checking the stability of the returned candidate minimizers; 
\item rerunning once with a tighter tolerance $t$ and/or a slightly larger degree and verifying that no additional minimizers appear beyond the target precision; 
\item refining each candidate by a local descent method applied to the original objective $f$ (using only evaluations of $\Gamma$ when derivatives are unavailable) and checking that the finite-difference gradient is small; 
\item when second-order information is accessible or can be estimated, checking that the Hessian is positive definite at the refined candidate.
\end{itemize}
\end{rem}
}

\section{Implementation choices}

In this section, we start to explain some design choices in our implementation,
motivated by the results in~\cite{OptimalWeighted}, which we actually use in the
proof of \cref{thm:main:noise}. The next subsection gives an overview of the
algebraic methods we use to solve systems of polynomial equations and which are
well-established in computer algebra.

\subsection{Details on \dlsp}\label{ssec:dlsp}

We now describe the more practical aspects of the implementation
of~\dlsp, how we compute the polynomial least-squares
approximant $w_{d, \Ss}$, as defined in~\eqref{eq:discrete_LS}.  One of the
pitfalls of working with high-degree polynomials is that, if not handled
properly, the conditioning of any computation can rapidly become problematic,
hence our focus on the numerical stability in the construction of $w_{d, \Ss}$.
The choice of optimal sampling measure $\mu$ is crucial in this regard, and for
two reasons, it defines the orthonormal basis of $\Pscr_{n, d}$ in which we
compute the approximant, and it defines the quantity which
sets the lower bound on the number of samples we have to consider (see
Steps~\ref{step:k} and~\ref{step:S}).   
It is the distribution of the sample set $\Ss$ in $\Cs$, coupled with
the choice of orthonormal basis on $\Pscr_{n, d}$ that ensures numerical
stability in the computation of the approximant $w_{d, \Ss}$.

The choice on an optimal sampling measure $\mu$ does not have a straightforward
answer, it is dependent on the objective $f$ we are seeking to approximate.  How
well suited a measure $\mu$ is for approximating $f$, or a family of objective
functions over $\Cs$ is quantified by the associated Lebesgue constant. The
Lebesgue constant is in turn defined by the orthonormal basis induces by $\mu$
on the space of polynomials and the function $f$ itself.  In certain cases, it
can be computed or estimated, see for instance~\cite{LebesgueConstant}, but more
general statements are hard to make, our choices are more empirically based.  In
the univariate case, working over uniformly spaced interpolation points is far
from optimal, instead, Chebyshev nodes are a much preferred option. They are nearly
optimal for univariate polynomial interpolation~\cite{Lagrange_interp}, see
also~\cite{Cheney,rivlin1974chebyshev,upmc.tref}.  This observation extends to
the multivariate, over-constrained construction of $w_{d, \Ss}$, and the
weighted least-squares method for polynomial approximations introduced
in~\cite{OptimalWeighted} is well suited to take advantage of that.  For that
reason, we replace the probabilistic construction of $\Ss$ by a deterministic
one, where $\Ss$ is the tensorized Chebyshev grid of cardinality $\samplesize$,
\begin{equation} \label{eq:grid_chebyshev} \Ss = \underbrace{S \times S \times
        \cdots \times S}_{n}, \quad with \quad S = \left\{\cos\left(\frac{(2j +
                1)\pi}{2\left\lceil \sqrt[n]{k} \right\rceil}\right)\mid j = 0,
                \ldots, \left\lceil \sqrt[n]{k} \right\rceil\right\}.
            \end{equation} In other words, $\Ss$ is a discretized analogue of
            size $\samplesize$ of the density function associated to the tensorized
            Chebyshev measure $\mu$ on $\Cs$.  The induced orthogonal basis
            $\{\psi_1, \ldots, \psi_m\}$ is obtained by applying the following
            transformation of the standard monomial basis of $\Pscr_{n, d}$:
            \begin{equation}\label{map:change_basis} x_1^{\nu_1}\ldots
            x_n^{\nu_n} \rightarrow T_{\nu_1}(x_1)\ldots T_{\nu_n}(x_n),
        \end{equation} where $T_i$ denotes the Chebyshev polynomial of the first
        kind of degree $i$. 
        The study of the geometric properties of varieties parametrized by
        tensorized Chebyshev polynomials has only recently begun, see~\cite{ChebyshevVarieties}. 
        Furthermore, the tensorized Chebyshev measure $\mu$
        minimizes the size of the sample set, compared to other choices of
        measures such as the uniform or the Gaussian measure, without
        necessitating a weighted least-squares formulation to attain optimal
        convergence rates, see~\cite[Section 6]{OptimalWeighted}. Although the
        number of samples does not significantly affect the complexity of our
        algorithm, we still benefit from using a reasonable number of them.  The
        coefficients of the approximant $w_{d, \Ss}$ are then computed by
        solving the linear system of normal equations \begin{equation}
            \label{eq:linear_system} L^T L x = L^T F, \end{equation} where $L$
            is the Vandermonde-like matrix and $F$ the associated vector of
            evaluations: \begin{equation*} L = \begin{bmatrix} \psi_{1}(s_1) &
                \ldots & \psi_{r}(s_1)\\ \vdots & & \vdots\\ \psi_{1}(s_k) &
            \ldots & \psi_{r}(s_k) \end{bmatrix}, \quad F = \left[f(s_1), \ldots
        , f(s_k)\right]^T, \quad \text{with } s_{i}\in \Ss \text{ and } r =
    \binom{n+d}{n}.  \end{equation*} The numerical stability in solving the
    linear system~\eqref{eq:linear_system} depends on the condition number of
    the Gram matrix $L^T L$, and that is where the orthogonality of $\{\psi_1,
    \ldots, \psi_m\}$ with respect to the measure $\mu$ comes into play, as it
    ensures that the Gram matrix stays well-conditioned.  For a reasonable
    choice of degree $d$, we have observed that 
    the vector of coefficients $c$ obtained by
    solving~\eqref{eq:linear_system} in standard double-precision floating point
    arithmetic is sufficiently accurate to define the near-best polynomial
    approximant $w_{d, \Ss}$.  
    Once a potential polynomial approximant has been identified,
    we convert the floating point vector of coefficients $c$
    into a vector of rational numbers. We reverse the
    transformation~\eqref{map:change_basis} in exact arithmetic to expand the approximant $w_{d, \Ss}$ in the standard monomial basis. This is a computational cost we have to pay to guarantee that we do preserve the quality of the approximant we just computed. This is also where we see the potential for future applications of geometric results on varieties parametrized by Chebyshev polynomials of the kind of \cite[Theorem 6.2]{ChebyshevVarieties}.
    Last but not least, we initiate the resolution of the system of polynomial equations defined by the vanishing of its
    partial derivatives, in exact arithmetic.  
\subsection{Details on \posso}\label{ssec:posso}

We rely on algebraic methods for solving systems of polynomial equations, in
particular on Gröbner bases algorithms (see e.g. \cite{Cox}). Compared to
numerical or semi-numerical methods such as homotopy continuation, these 
methods have the advantage to guarantee that all solutions are computed, without
suffering from numerical accuracy or conditioning issues.  

On input a system of polynomial equations $f_1= \cdots= f_s = 0$ in $\QQ[x_1,
\ldots, x_n]$, Gröbner bases algorithms compute an equivalent system which,
through a multivariate division algorithm depending on an admissible order on
the monomials, allows to decide whether a given polynomial lies in the ideal
generated by the input equations. This fundamental property allows one to
compute "modulo" the solutions of the input system (taking into account
multiplicities).  
We will not delve into the details of the algorithmic of Gröbner bases, but
just mention that modern algorithms for computing Gröbner bases (see e.g.
\cite{F4, F5}) boils down to row echelon form computations of matrices with
columns indexed by monomials up to some degree $\bm{d}$ (and sorted according to some
admissible monomial ordering) and rows which are multiples of polynomials which
are algebraic combinations of the
input equations which are multiplied by monomials in
order to reach that degree. This core machinery comes with criteria determining
when this degree is large enough. Generically, when using a graded monomial
ordering (sorting the monomials first by degree), the input equations have
degree $d$, the maximum degree up to which one needs to generate such matrices
is  $1+n(d-1)$ (see \cite{Lazard}). 

When the input equations have finitely many complex solutions, and again, under
some genericity assumption, one can obtain from such a Gr\"obner basis, a rational
parametrization of the solution set as follows: 
\[
    w(x_n) = 0, \quad x_{n-1} = \frac{v_{n-1}(x_n)}{w'(x_n)}, \ldots, 
    x_1 = \frac{v_1(x_n)}{w'(x_n)}
\]
where the $v_i$'s and  $w$ are univariate polynomial. Again, this is done
through linear algebra computations (see e.g. \cite{FAUGERE1993329, BeNeSa22}). 
Once such a rational parametrization is obtained, the real solutions to the
input system of equations are extracted through real root isolation and
evaluation algorithms. 

\section{Software implementation and practical experiments}
\label{sec:numerical}
We re-iterate that
our end goal is to solve Problem~\ref{pbm} by finding \emph{all} local
minimizers of the objective function $f$, satisfying some regularity
assumptions, over the interior of the compact domain $\Cs$.

Hence, the interest of our 
method is not at competing over computation
timings with the state of the art numerical global optimization methods of the
likes of particle swarm optimization (see~\cite{particle_openFPM, ParticleSwarmOptimization}), but
more on its "global feature" meaning that, under regularity assumptions, 
it allows to grab \emph{all minimizers}. However, we will see below that our
algorithm performs sometimes more efficiently than purely numerical methods.  

\paragraph*{Computing platform and Software infrastructure.}
All the computations we report on are performed on a computing server equipped
with 
Intel(R) Xeon(R) Gold 6244 CPU @ 3.60GHz CPUs 
yielding $32$ threads for parallel computations and with  $1.48$ TBytes of
RAM. 

We implemented our algorithm using the
\href{https://julialang.org/}{$\mathsf{Julia}$} programming language, within
the \href{https://github.com/gescholt/Globtim.jl}{$\mathsf{Globtim}$} library.
The core of this implementation yields functionalities for providing polynomial
approximants of the objective function to optimize. The step solving systems of
polynomial equations is handled by calling (through a file interface) 
to the \msolve library
\cite{BES21}, which is a high-performance computing library based on algebraic
methods (Gröbner bases) for solving such systems. 
The \href{https://github.com/gescholt/Globtim.jl}{$\mathsf{Globtim}$} library
also allows the use of
\href{https://www.juliahomotopycontinuation.org/}{$\mathsf{HomotopyContinuation.jl}$} 
    which implements numerical homotopy continuation routines for solving
    polynomial systems. Below, we report on experiments where \msolve is used
    exclusively as we observed that when the degree of our polynomial
    approximants tends to increase, the numerical  
\href{https://www.juliahomotopycontinuation.org/}{$\mathsf{HomotopyContinuation.jl}$} 
may miss some solutions. Note also that \msolve allows for parallel computing ;
when we use this feature, we mention it as well as the number of used threads. 
We also report on results when initiating local  numerical optimization routines such
as the Broyden–Fletcher–Goldfarb–Shanno (BFGS) algorithm 
with the local minimizers provided by our
method as input, solving the local minimization problem, hence, refining our
numerical approximations. 

We compare our method with \chebfun. As already commented in \cref{sec:intro},
the unique implementation we are aware of which does solve \cref{pbm} is
\chebfun, in the specific case where $n=2$, hence using \chebfuntwo. Recall
also that \chebfunthree does not solve \cref{pbm} but focuses on minimizing the
objective function. Both rely on polynomial approximants but with a specific form
since they are sums of products of \emph{univariate} polynomials (see again
\cref{sec:intro}). We specifically
compare the sizes/degrees of the polynomial approximants which are produced by
\chebfun with ours. 

\paragraph*{Methodology and performance evaluation.}
From \cref{sec:algos}, there are two correlated phenomena to consider.
As we increase the degree $d$, we first look at the convergence of the orthogonal $\LL^2$-projection of $f$ onto $\Pscr_{n, d}$. This convergence of the approximant $w_{d, \Ss}$ towards the objective function $f$ in the discrete $\LL^2$-norm is observed numerically and has been the topic of extensive study in~\cite{MIGLIORATI}. The second observation is the convergence at the level of the critical points and local minimizers of $w_{d, \Ss}$ towards their objective function's respective counterparts.  

Hence, the separation distance between the critical points of the objective function $f$ and those the approximant $w_{d, \Ss}$ will be an indicative metric in most of the examples we present.
This numerical estimate can only be as good as the accuracy we have on the critical points of the objective function $f$, which we denote $\crit(f)$.
In examples where it is relatively easy to isolate $\lmin(f)$, we will also use the separating distance between $\lmin(f)$ and the critical points of $w_{d, \Ss}$ as a metric of accuracy.
To compute approximation errors, we use a discrete $\LL^2$-norm, meaning a Riemann sum over the structured grid of sample points $\Ss$.
Recall that $\Ss$ is constructed as the tensor of a collection $\left(x_i\right)$ for $i = 1,\ldots, k$ of points distributed according to the measure $\mu$ of our choice.  
\begin{equation}
    \label{def_grid}
    \Ss = \left(x_i\right)_{i=1}^k \otimes \cdots \otimes \left(x_i\right)_{i=1}^k.
\end{equation}
For consistency across the examples, we take $\Ss$ to be a grid of $120^2 = 14400$ points distributed according to the distribution $\mu$, meaning that it is either a uniform grid in the Legendre case or tensorized Chebyshev nodes when using orthogonal Chebyshev polynomials.
We do adjust the dimensions of the domain we are working over, by translating and scaling the box $\Cs$.
Accordingly, we have to adjust the tolerance with which we consider a critical points "captured", this according to the objective function $f$.
To have a good estimate of the error we attain in the $L^2$-norm, we consider $\Ss$ as a grid, with each cell $c \in \Ss$ defined as $c = \prod_{j=1}^n I_j$, where each $I_j$ is an interval determined by adjacent points in $\left(x_i\right)$ for $i = 1,\ldots, k$. Let $x_c$ denote the midpoint of cell $c$, and $\text{vol}_c$ denote its volume with respect to the measure $\mu$.
The discrete $\LL^2$-norm of the error according to the measure $\mu$ is then estimated via the Riemann sum:
\begin{equation}
    \label{eq:discrete_L2}
    err_{d,\Ss}= \sum_{c\in \Ss} \text{vol}_c \left(f(x_c)-w_{d,\Ss}(x_c)\right)^2.
\end{equation}
In practice $f(x_c)$ is a numerical evaluation, potentially subject to noise and $w_{d,\Ss}$ the DLSP of degree $d$.

\paragraph*{Test-suite and overall results.}
The key parameter to study in order to measure how well-designed is our
algorithm, is the minimum degree at which our polynomial approximants can
capture the searched minimizers of our objective function in $\Cs$ and how
close it is to "the best possible" polynomial approximant of the same degree. 
This is of course a difficult task if one does not have a priori an accurate
knowledge of regularity properties of the objective function. 

Still, for this task, a relevant family of examples are objective functions
which are polynomials. Indeed, when the objective function is a polynomial of
degree $d$ (generic enough so that its gradient ideal is radical and has
finitely many complex roots), one can measure if our polynomial approximants of
degree $d$ allow us to capture with a reasonable accuracy all the minimizers
of the objective function (which we can compute independently). 
This is what we do, with polynomial objective functions  in the
three-dimensional case.

\medskip
\emph{We observe in \cref{exm:polynomial} that, indeed, when applying it to objective polynomial functions of degree $d\in \{4, 6, 8\}$ our method yields polynomial
approximants of degree $d$ which capture all its local minimizers at high accuracy.} 

\medskip
Of course, we consider a second set of examples, which are \emph{not}
polynomial, considered as standard difficult benchmarks in global 
optimization~\cite{Jamil2013} as well as  Problem~4 of the 100-Digit
challenge~\cite{100DigitChallenge}, focusing on the case $n=2$, hence allowing
fair comparisons with \chebfuntwo, on the one hand, and an empirical analysis of
the quality of the obtained results (through plotting).
For those benchmarks where the local minimizers need to be computed in a box
which is not $[-1, 1]^2$, we proceed by scaling to retrieve the case  $[-1,
1]^2$. 
We detail below the results obtained on the De Jong function as
well as Hölder's table function 2 from~\cite{Jamil2013} (see
\cref{exm:dejong} and \cref{exm:holder_table}). 

\medskip
\emph{We observe that our method allows to capture all minimizers of these objective function using polynomial approximants of degree way smaller than
the ones used by \chebfuntwo, which either misses some of them or meets scaling
issues, typically on the Hölder's table function 2.}
\medskip 

Our last example is the composite function obtained by adding two copies of the
Deuflhard function (see \cref{exm:deuflhard_4d}). We show how to combine our
approach with a subdivision one to solve \cref{pbm} on this problem.

Finally, we emphasize that in all cases we have tested, the number of complex
solutions to the critical point system obtained from a polynomial approximant of
degree $d$ is  $(d-1)^n$. This parameter governs the difficulty of solving
systems of polynomial equations, and hence indicates how difficult it is the
solving of \cref{pbm} when $n$ increases. 

\paragraph*{Detailed numerical results.} We provide now details on the behavior
of our algorithm and, when it is possible comparisons with \chebfuntwo. 
\begin{exm}
    \label{exm:polynomial}
    We start with polynomial objective functions of degree $d\in \{4, 6, 8\}$
    in the case  $n=3$ which are chosen with coefficients picked randomly. 
    These polynomials are dense. 
    To identify when the minimizers of such a polynomial $f$ are captured, 
    we use \msolve to compute all its critical points a priori, solving the
    system   
    \begin{equation*}
        \pder{f}{x_1}= \ldots = \pder{f}{x_n}=0.
    \end{equation*}
    For all the above degrees, the experiment is repeated $4$ times. 
    We also run our experiments in the noisy and the noiseless model since the
    latter one also makes sense for such objective functions. 
    To validate our choice of working with the Chebyshev tensorized polynomial
    basis, we also compare the results obtained using the Legendre tensorized
    polynomial basis. 
    We construct $w_{d, \Ss}$ in both the Chebyshev and
    Legendre tensorized polynomial basis, compute their critical points and
    compare their convergence rates towards the critical points of the objective
    functions, see Figure~\ref{fig:dim_3}.
    \begin{figure}[htbp]
        \centering
        \begin{subfigure}[b]{0.35\textwidth}
            \centering
            \includegraphics[width=\textwidth]{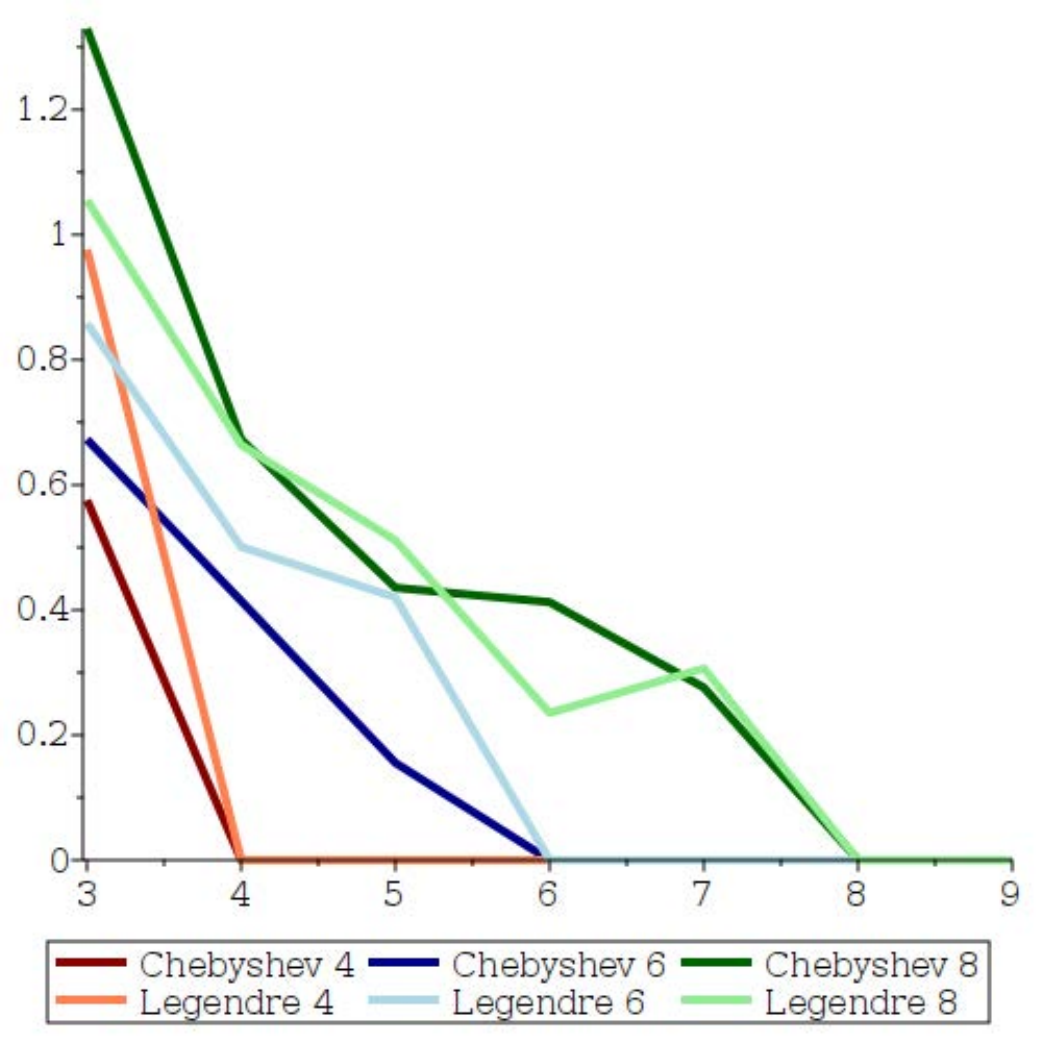}
            \caption{Average distance noiseless}
        \end{subfigure}
        \begin{subfigure}[b]{0.35\textwidth}
            \centering
            \includegraphics[width=\textwidth]{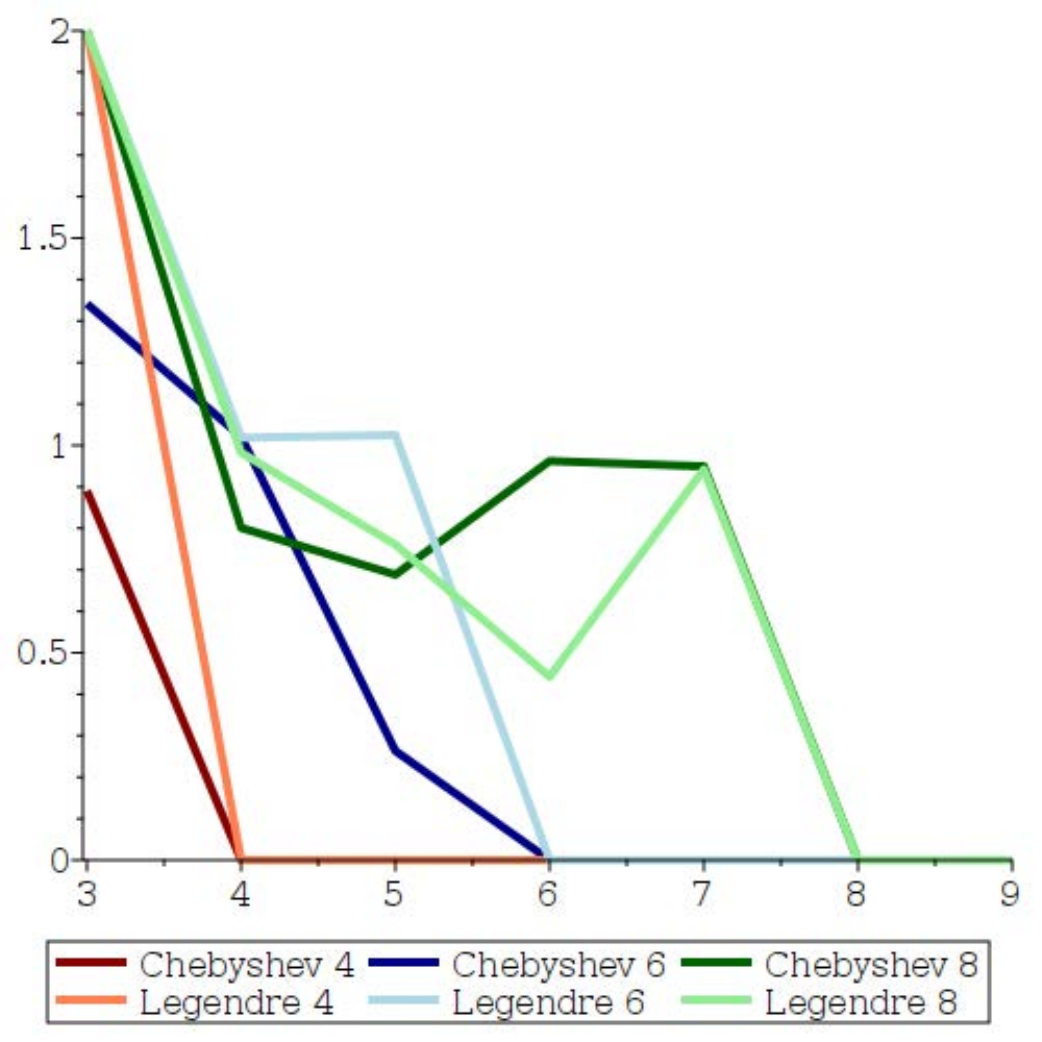}
            \caption{Maximal distance noiseless}
        \end{subfigure}
        \\
        \begin{subfigure}[b]{0.35\textwidth}
            \centering
            \includegraphics[width=\textwidth]{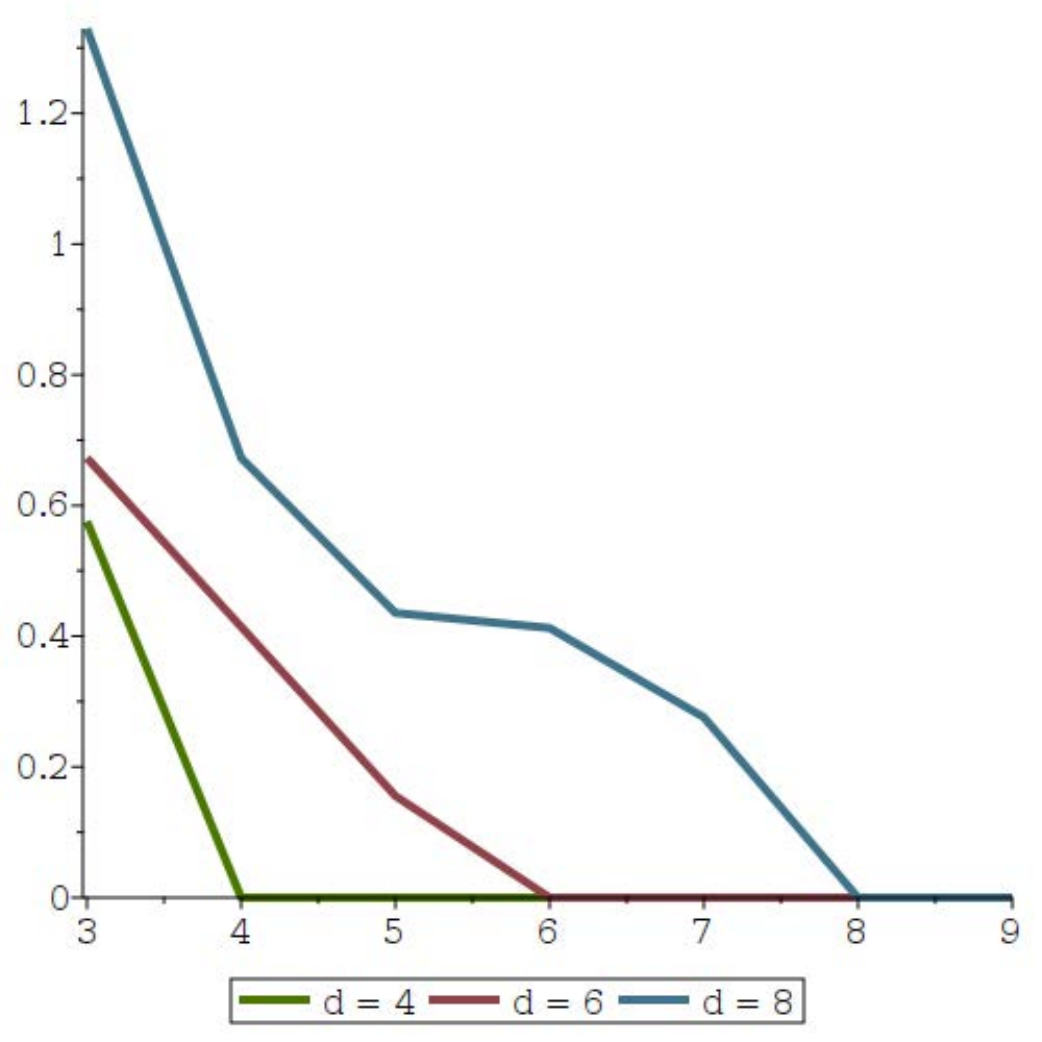}
            \caption{Average distance noisy}
            \label{fig:subfig_noisy_n3_avg}
        \end{subfigure}
        \begin{subfigure}[b]{0.35\textwidth}
            \centering
            \includegraphics[width=\textwidth]{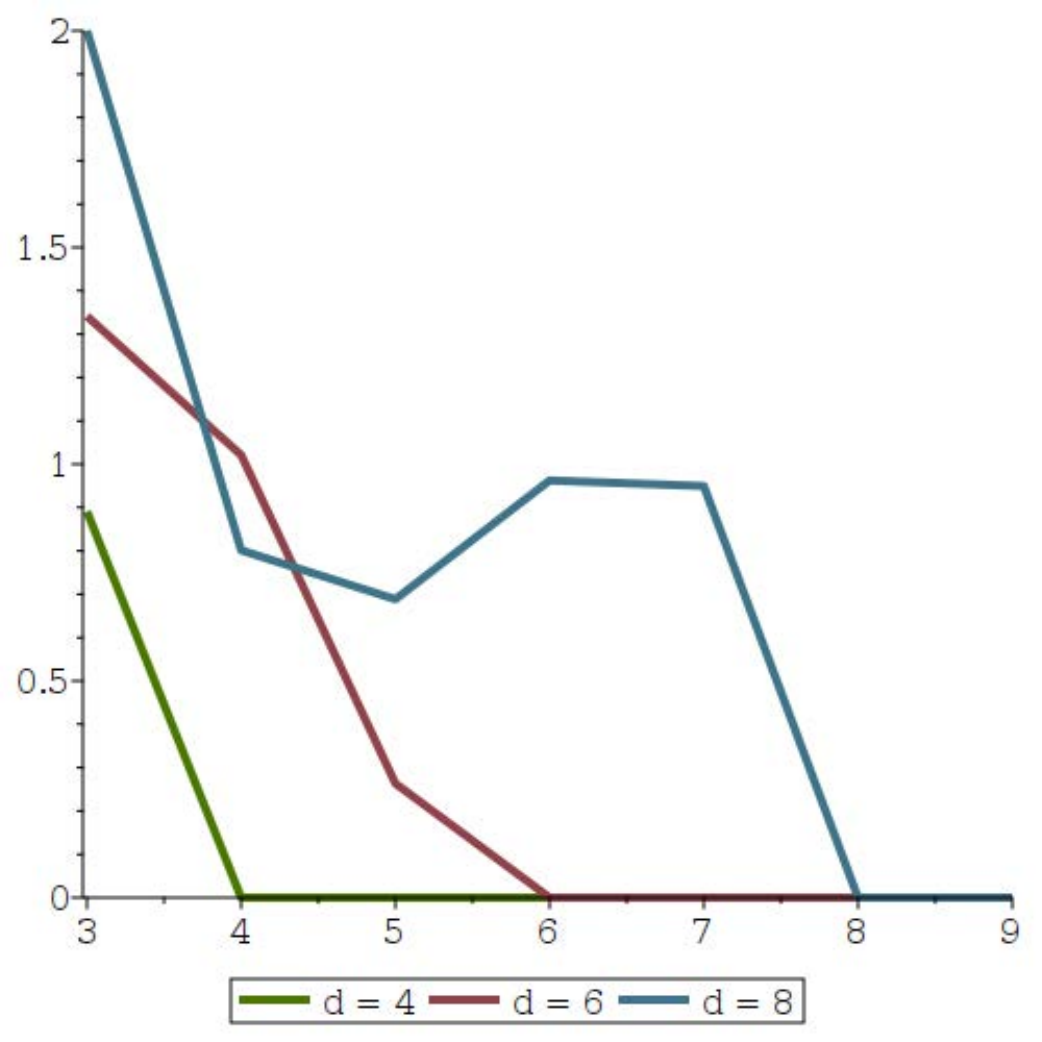}
            \caption{Maximal distance noisy}
            \label{fig:subfig_noisy_n3_max}
        \end{subfigure}
        \caption{Separating distance between critical points of the objective
        and approximant functions $w_{d, \Ss}$, constructed in the tensorized Chebyshev and the tensorized Legendre polynomial basis, of increasing degrees $d = \{3, \ldots, 9\}$. The experiment is repeated with noisy evaluations in the Chebyshev basis in figure~\ref{fig:subfig_noisy_n3_avg} and~\ref{fig:subfig_noisy_n3_max}.}
        \label{fig:dim_3}
    \end{figure}
    We observe in Figure~\ref{fig:dim_3} that the polynomial approximant $w_{d, \Ss}$ always seems to capture all critical points of the objective function once it reaches the threshold degree $d$ where both the degree of the objective and the approximant are equal.
    We observe that their respective critical points attain a separation
    distance to the objective's on the order of $10^{-9}$. 
    As we approach that threshold degree $d$, we already observe that starting
    with approximants of degree $d-2$, the "largest" local extrema of the
    objective function get captured, usually at a distance on the order of
    $10^{-2}$ to $10^{-1}$, but this does not extend to all critical points of
    the objective function.
    We also observe that using Chebyshev tensorized polynomial basis yields
    better results for polynomial approximants of degree smaller than  $d$. 
    These experiments are done in the exact model.

    To simulate the noisy model, we introduce some random noise with magnitude
    bounded by $10^{-1}$. 
    In this case, we observe the very same behavior as in the noiseless
    case, see Figure~\ref{fig:dim_3}.
    Finally, it is worth to note that, on these examples, \chebfunthree's
    \texttt{max3} function returns a point on the boundary of $\Cs$, hence
    failing to solve the global optimization problem.  
\end{exm}

We now switch to objective functions which are not of polynomial type.

\begin{figure}[!hbtp]
    \centering
    \begin{subfigure}[]{0.35\textwidth}
        \centering
        \includegraphics[width=\textwidth]{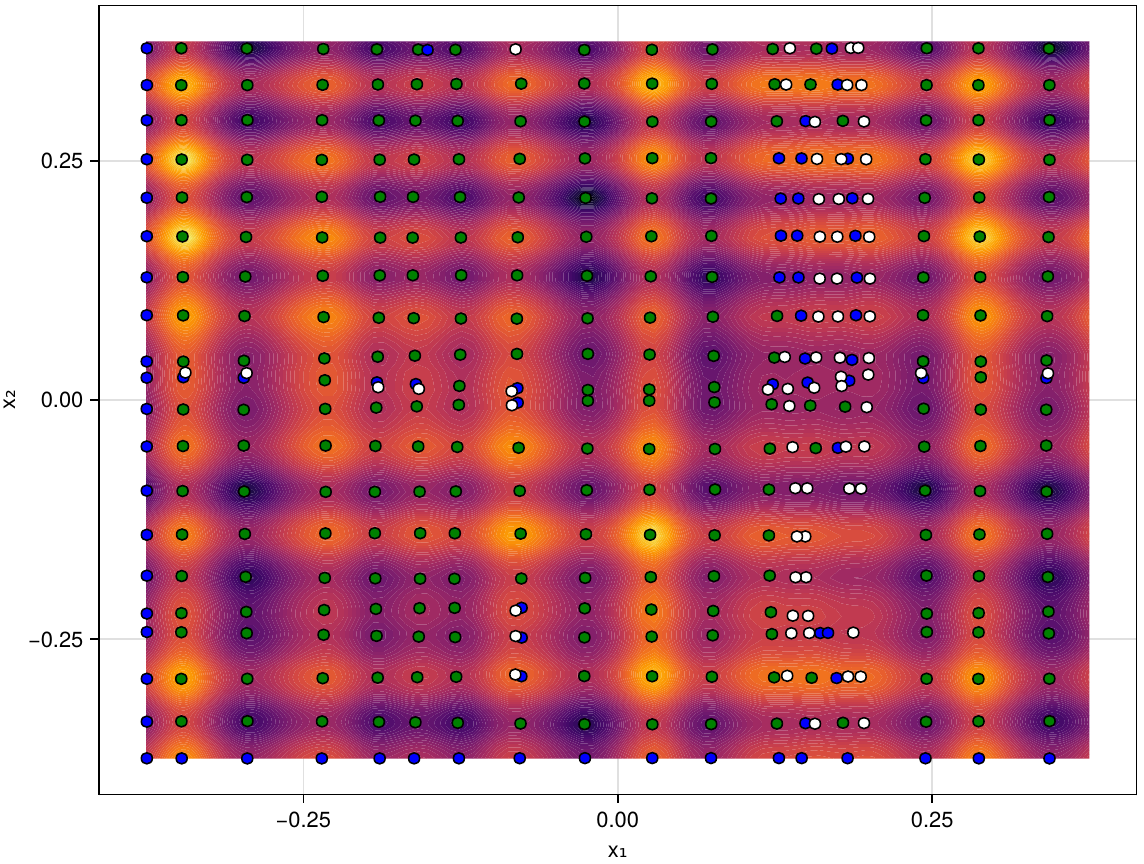}
        \caption{$d=34$}
        \label{fig:subfig_tref_34}
    \end{subfigure}
    \hfill
    \begin{subfigure}[]{0.35\textwidth}
        \centering
        \includegraphics[width=\textwidth]{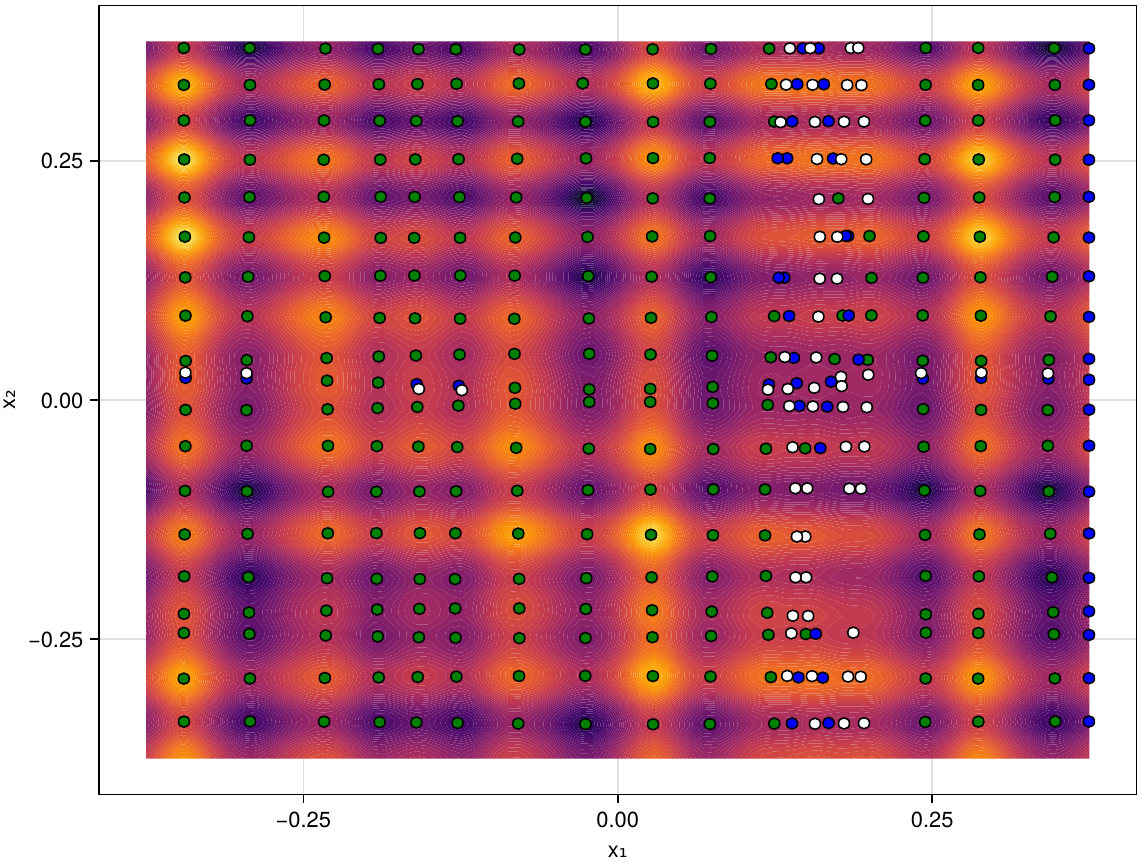}
        \caption{$d=36$}
        \label{fig:subfig_tref_36}
    \end{subfigure}
    
    \begin{subfigure}[b]{0.35\textwidth}
        \centering
        \includegraphics[width=\textwidth]{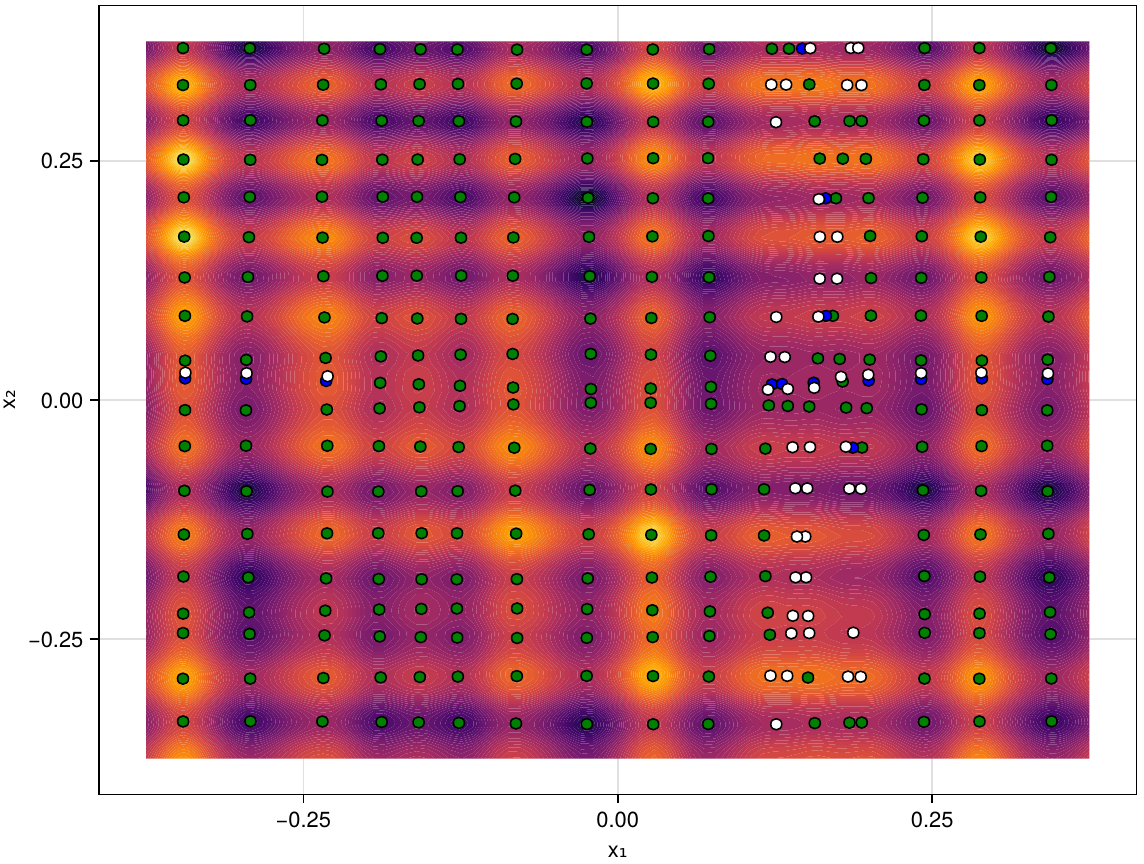}
        \caption{$d=40$}
        \label{fig:subfig_tref_40}
    \end{subfigure}
    \hfill
    \begin{subfigure}[]{0.35\textwidth}
        \vspace{-7cm}
        \centering
        \includegraphics[width=.5\textwidth]{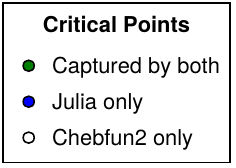}
        \label{fig:legend}
    \end{subfigure}
    \caption{The critical points of the approximant $w_{d, \Ss}$ computed with
    \msolve are plotted in green and blue. Green if they fall within a
distance of $5.e{-3}$ from a critical point computed with Chebfun2 and blue for
those not. The white points are the missing critical points computed by Chebfun2
but not \href{https://github.com/gescholt/Globtim.jl}{$\mathsf{Globtim}$}.}
    \label{fig:trefethen_all}
\end{figure}
\begin{exm}\label{exm:trefethen}        
    We consider a variant of~\cite[100 Digit Challenge]{100DigitChallenge} which
    consists in computing all local minimizers in $\Cs$ of 
    \begin{multline}
        \label{func:tref}
        f(x, y) = \exp(\sin(50x)) + \sin(60\exp(y)) + \sin(70\sin x) \\
        + \sin(\sin(80y)) - \sin(10(x + y)) + \frac{x^2 + y^2}{4}
    \end{multline}
    (and refine the solution to a $100$ digits of precision).
    Here, the objective function $f$ exhibits quite complex level sets, with $2720$ critical points of varying magnitudes, all irregularly distributed in $[-1, 1]^2$.
    Capturing all these points all at once would require a polynomial
    approximant of degree so high that the step calling \posso
    would run for too long. Therefore we reduce the domain of approximation
    to $[-3/8, 3/8]^2$. 

Using \chebfuntwo, we compute $330$ critical points of $f$ in the domain $[-3/8, 3/8]^2$ via the construction of a rank 4 polynomial tensor of degree 53 in $x_1$ and 403 in $x_2$.
We compare this against the outputs of our algorithm for degrees $d = \{34, 36,
40\}$ (hence, a much smaller degree, see Figure~\ref{fig:trefethen_all}). 
We can observe that while the method effectively captures the most prominent local minimizers, it encounters difficulties with smaller ones, particularly in the strip $0.1 \leq x_1 \leq 0.25$ where function fluctuations are minimal. Higher degrees improve accuracy, as we see in Figure~\ref{fig:trefethen_all}. 
In this example it may be hard to quantify the general trend of capturing more critical points as the degree increases, some critical points are found, some appear lost. 
in the present case, we can tell we are limited in accuracy in the present construction even before computing the critical points: the discrete $\LL^2$-norm of the error of approximation when computing the approximant $w_{d, \Ss}$ shows progressive improvement as the degree $d$ increases, but remains high.
\begin{center}
    \begin{tabular}{|c|c|c|}
        \hline
        d & $err_{d, \Ss}$ & Critical points in bounds \\
        \hline
        34 & $3.12e{-1}$ & 332\\
        36 & $2.17e{-1}$ & 317\\
        40 & $7.70e{-2}$ & 290\\
        \hline
    \end{tabular}    
\end{center}
\end{exm}

\begin{exm}
    \label{exm:dejong}
    We consider here as an objective function the De Jong
    function~\cite{Jamil2013}, defined on $[-50,
    50]^2$ which is a standard benchmark function in global
    optimization.
    Using \chebfuntwo, we construct an approximant of the De Jong function and
    compute its critical points using \chebfuntwo's $\mathsf{root}$ 
    function, restricted to
    the domain $[-50, 50]^2$.

    \begin{figure}[h]
        \centering
        \subfloat[The De Jong function]{
            \includegraphics[width=0.4\textwidth]{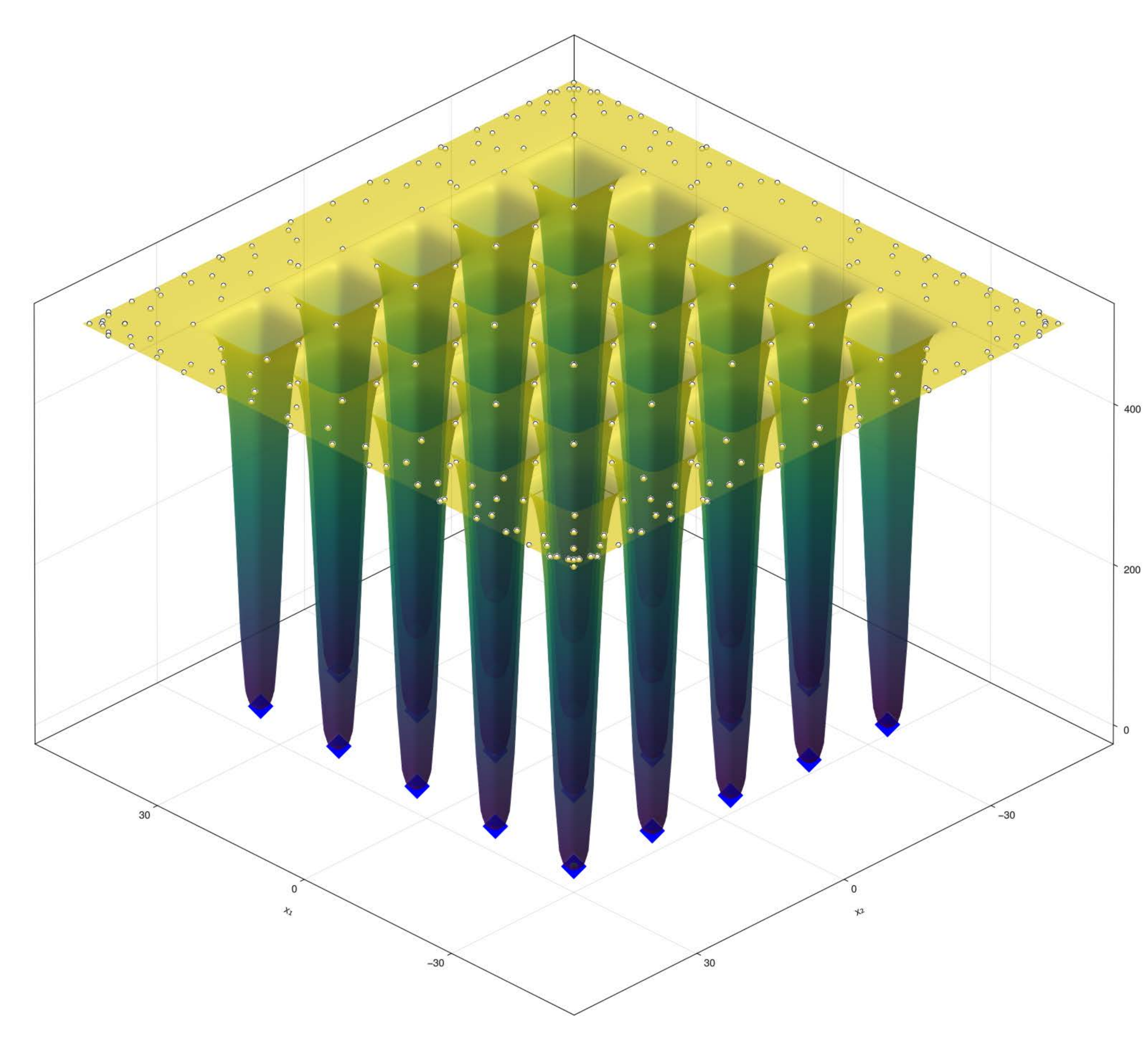}
            \label{fig:dejong_3d}
        }
        \hfill
        \subfloat[Vanishing locus of the partial derivatives of the Chebfun2 approximant with the computed roots in red.]{
            \includegraphics[width=0.55\textwidth]{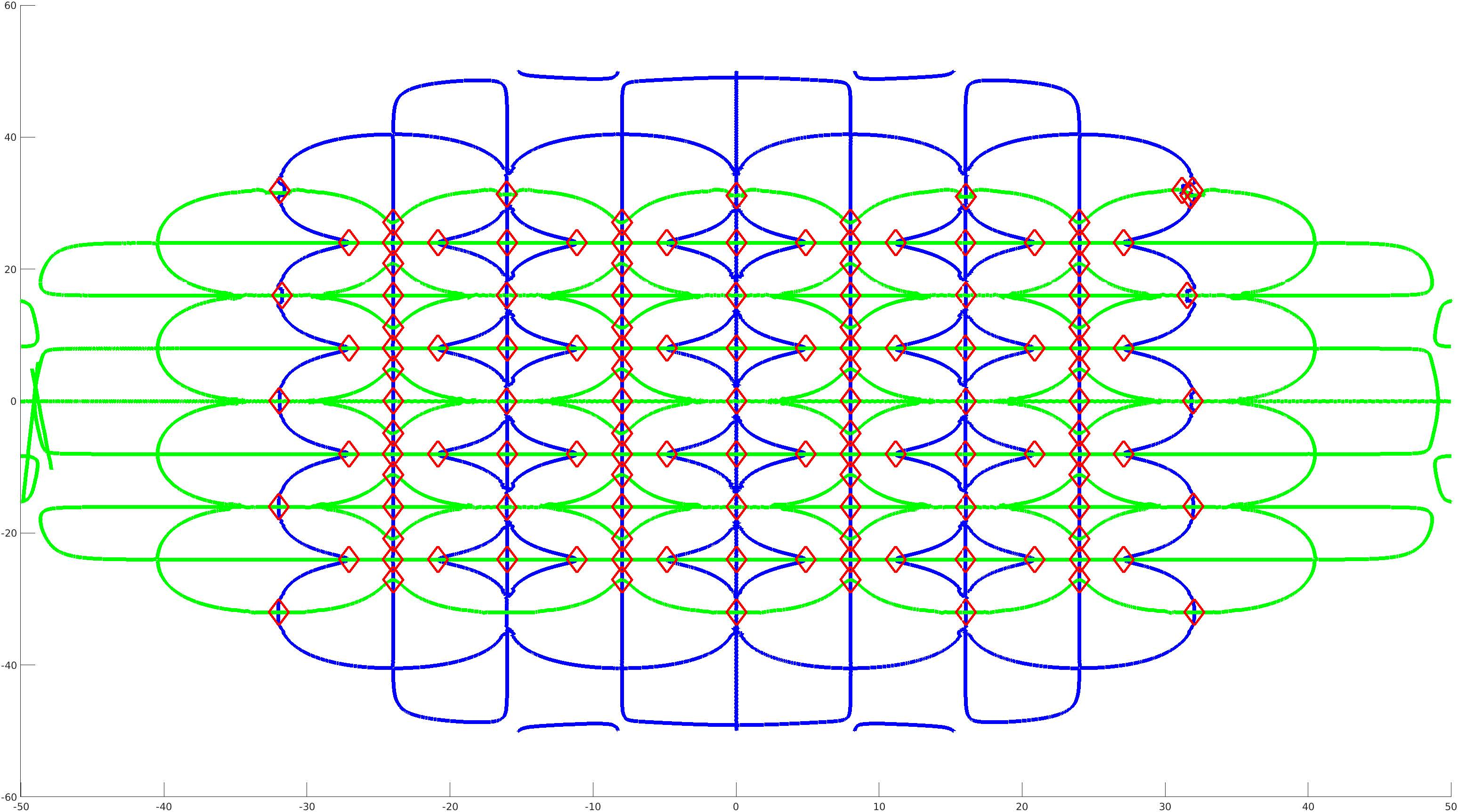}
            \label{fig:chebfun_crit}
        }
        \caption{The vanishing curves of the partials of the partials of the De
        Jong function computed with Chebfun2, in green and blue. The solutions found with \chebfun's polynomial system solver are plotted with red diamonds.}
        \label{fig:Object_dejong}
    \end{figure}
    As one is able to observe by comparing Figure~\ref{fig:dejong_3d}
    and~\ref{fig:chebfun_crit}, the local minimum attained around the point
    $(x, y)  = (-16, -31)$ is \emph{missing} from the set of critical points recovered
    by the $\mathsf{root}$ function in \chebfuntwo. 
    On the other hand, a triplet of critical points is found in the top right
    corner, in the vicinity of the point $(31, 31)$.  
    We think that both of these phenomena are \emph{numerical artifacts} due to the
    \emph{high degree} of the approximant constructed by \chebfuntwo, which in turn leads
    to some numerical instability in the computations of the gradient of the
    approximant. 
    Nevertheless, the approximant captures the objective function very
    accurately in $\LL^{\infty}$-norm. 
    This level of precision comes at the cost of a high degree polynomial
    approximant. 
    In the present case, \chebfuntwo constructs an approximant of the De Jong
    function of rank $91$, meaning it approximates the objective function using
    a sum of $91$ separable polynomial functions.

    \emph{Each of those polynomials in this construction is the product of two univariate polynomials of degree $d=1248$.} 

    Such a high-degree approximant is required because the approximation is
    constructed to be accurate to machine precision levels in the
    $\LL^{\infty}$-norm.
    In the context of computing the local minimizer of the objective, this is
    clearly counterproductive.
    The outputs of our least-squares approximation method capture the totality
    of the local minimizers of the De Jong function with a \emph{polynomial
    approximant of degree $d=20$ only}, see
    Figure~\ref{fig:dejong_deg_20}.

    \begin{figure}[ht]
        \centering
        \subfloat[$w_{20}$]{
            \includegraphics[width=0.45\textwidth]{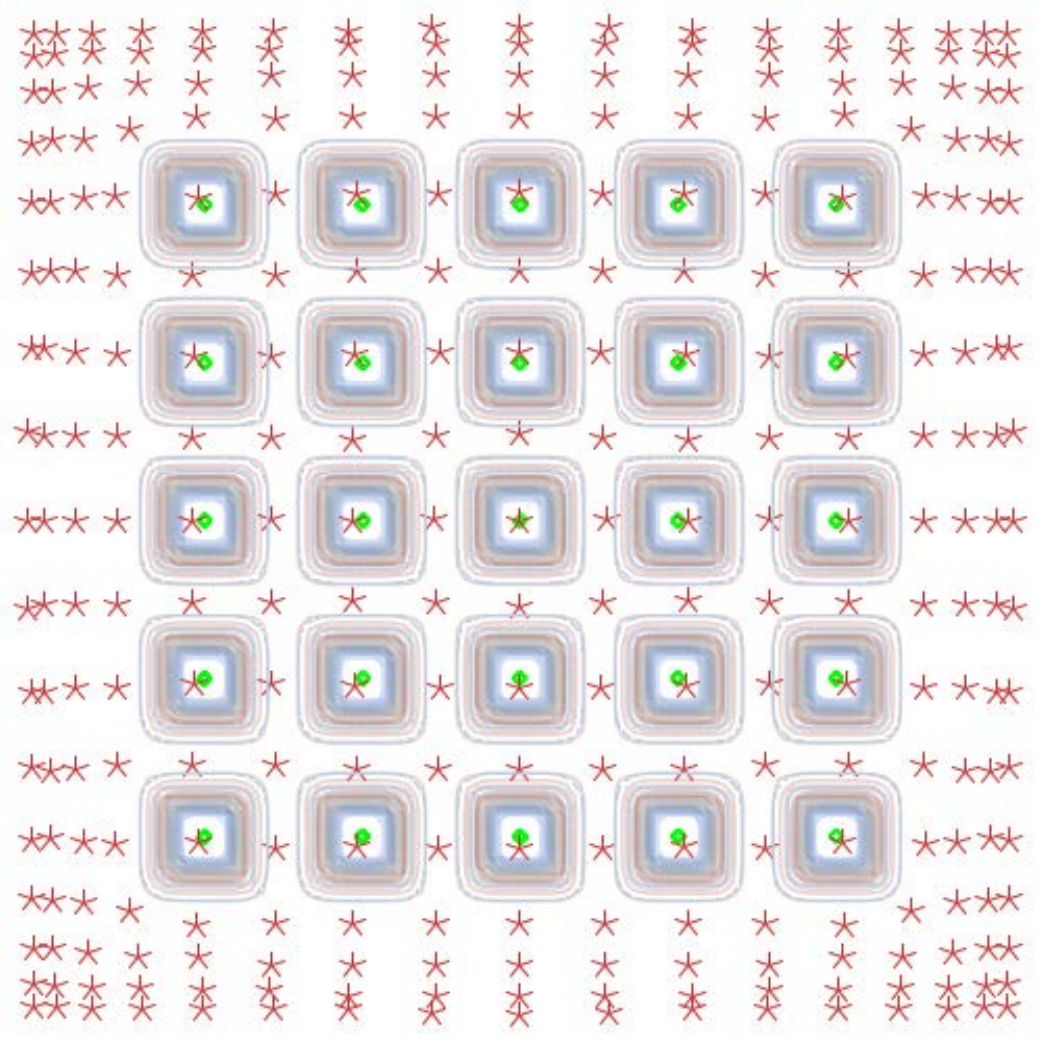}
            \label{fig:dejong_deg_20}
        }
        \hfill
        \subfloat[ Sum of the smallest separating distances from the local minimizers of the De Jong function to the critical points of $w_d$.]{
            \includegraphics[width=0.45\textwidth]{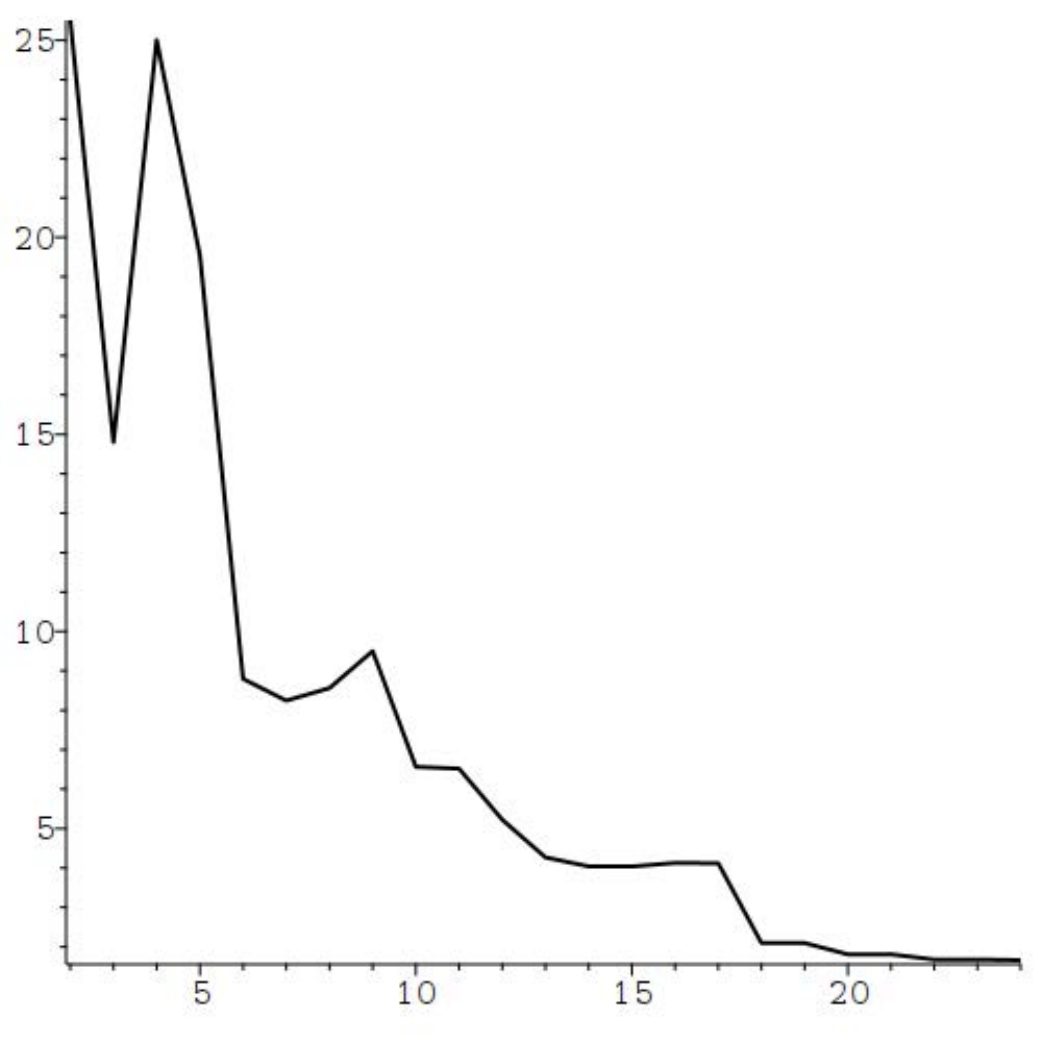}
            \label{fig:dist_dejong}
        }
        \caption{Starting at degree $d=20$, the approximant $w_{d}$ fully
        captures the local minimizers of the De Jong function over the domain $[-50,50]^2$, although it also adds many spurious critical points around the boundary. As the degree $d$ increases from $2$ to $24$, we observe convergence of the critical points of $w_{d, \Ss}$ towards the critical points of $f$, which we computed using \chebfuntwo.}
    \label{fig:dejong_full_capture}
    \end{figure}
    The critical points of the approximant $w_{d, \Ss}$ are computed with
    {\msolve}~\cite{BES21} and plotted as red stars. 
    The green points are the local minimizers found by initiating \href{https://www.maplesoft.com/support/help/Maple/view.aspx?path=Optimization/NLPSolve}{\textsf{NLPSolve}} on the original objective function $f$ at each of
    the red stars. 
    While approximating $f$ globally over the whole domain $[-50, 50]^2$, we
    capture all local minimizers of the objective starting at degree $d=12$.
    Although the approximant does not match the objective function in shape at
    that degree, the set of critical points it generates is already sufficient
    to capture the whole set $\lmin(f)$. \\
    We also note that although the critical points of the approximant converge
    towards the critical points of the objective function, see
    Figure~\ref{fig:dist_dejong}, while the approximant is degree deficient, its
    critical points may capture a local minimum of the objective function at
    degree $d$ and then lose it at degree $d+1$.

    Because the local minimizers of $f$ are disposed in a regular grid pattern
    and while the approximant is degree deficient, what we observe is happening
    is that the critical points of $w_{d, \Ss}$ are attempting to interlace the
    local minimizers of the De Jong function; this becomes quite apparent when
    we increase the degree of the approximant to $d=14$, see
    Figure~\ref{fig:dejong_w14_interlace}.
    In that case, for every adjacent pair of local minimizers of the objective
    function, there is a critical point of $w_{14}$ located approximately in
    between them. 
    It is only at degree $d=20$ that the right structure on the critical points
    of the approximant starts to emerge, with a saddle point in between each
    adjacent pairs of local minimizers. 

        A general pattern we observe is that the critical points of the approximant
    first appear along the edges of the domain $\Cs$ and then, as the degree
    increases, they move towards the critical points of the objective located at
    the center of the domain, see Figure~\ref{fig:dejong_deg_12}.

    To increase our confidence in the quality of the approximant,  we go up to
    computing the critical points of $w_{24}$.
    We believe all the critical points we see in Figure~\ref{fig:chebfun_crit}
    may be due to an over fitting the objective function with a polynomial
    approximant of too high degree.
    \begin{figure}[ht]
        \centering
        \subfloat[Contour plot of $w_{14}$ and its critical points.]{
            \includegraphics[width=0.45\textwidth]{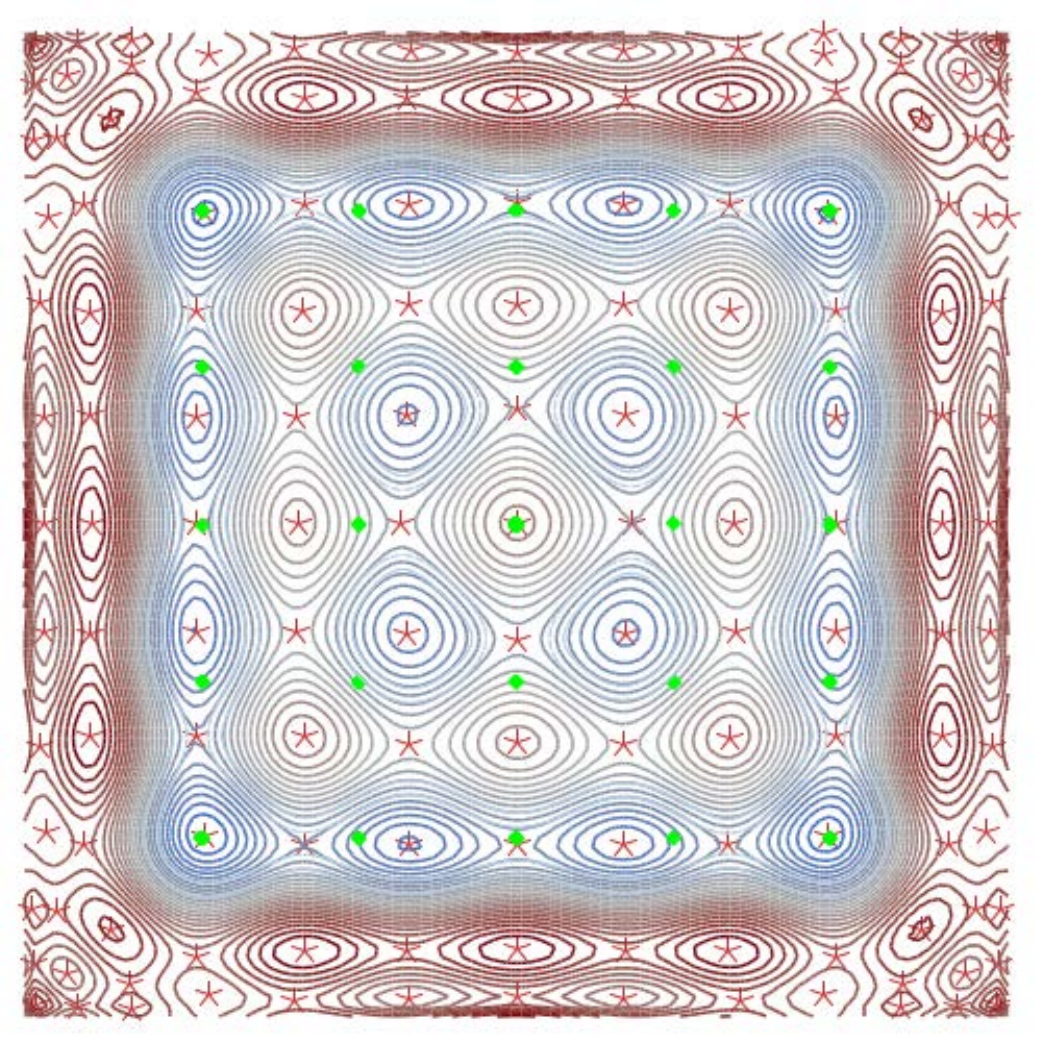}
            \label{fig:dejong_w14_interlace}
        }
        \hfill
        \subfloat[Contour plot of the De Jong function and the critical points of $w_{12}$.]{
            \includegraphics[width=0.45\textwidth]{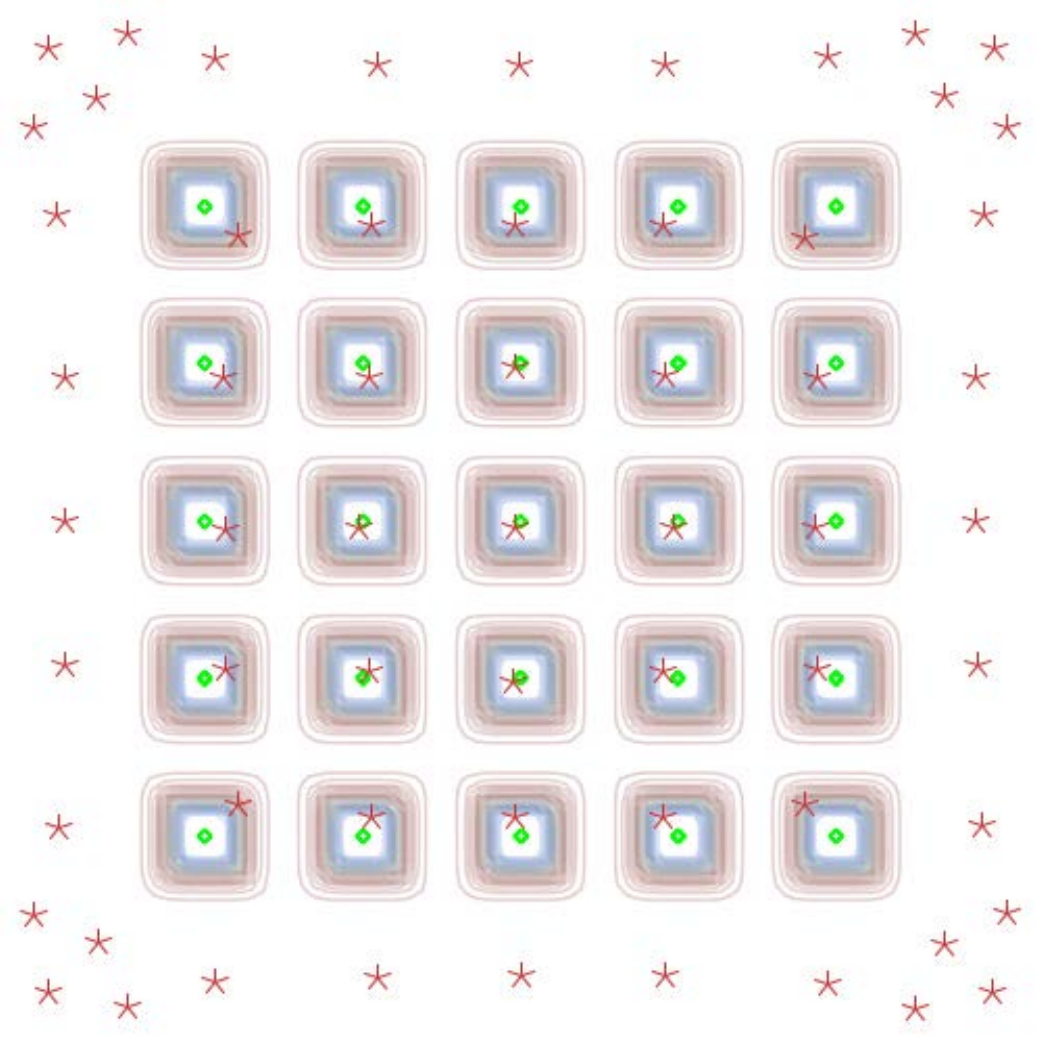}
            \label{fig:dejong_deg_12}
        }
        \caption{The $25$ local minimizers of the De Jong function are depicted
        in green, and the red stars represent the critical points of the approximant $w_{d, \Ss}$. Even though the shape of the approximant does not match the objective function, the critical points of the approximant are already capturing the local minimizers of the De Jong function at degree $d=12$.}
        \label{fig:contour_dejong_msolve}
    \end{figure}
 
\end{exm}

\begin{exm}[Deuflhard Function]\label{exm:deuflhard}
    The level sets of the Deuflhard function 
    \begin{equation}\label{eq:deufl}
    f(x, y) = (\exp(x^2 + y^2) - 3)^2 + (x + y - \sin(3(x + y)))^2    
    \end{equation}
    on the domain $[-1.1, 1.1]^2$ exhibit some quite interesting features.
    \begin{figure}[!htbp]
        \centering
        \begin{subfigure}[b]{0.35\textwidth}
            \includegraphics[width=\textwidth]{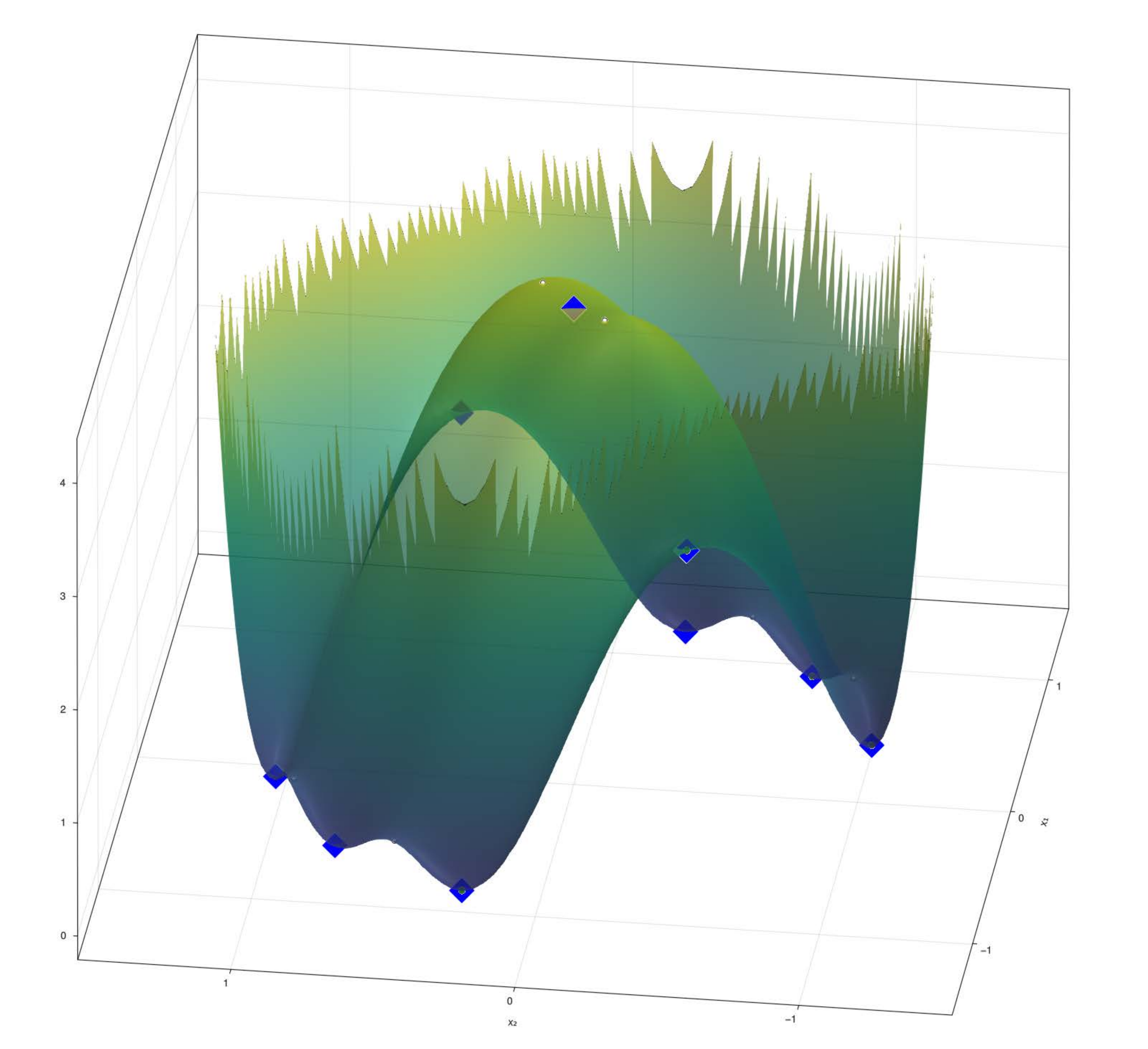}
            \caption{}
            \label{fig:subfig_a}
            \end{subfigure}
                \hfill
            \begin{subfigure}[b]{0.35\textwidth}
                \centering
                \includegraphics[width=\textwidth]{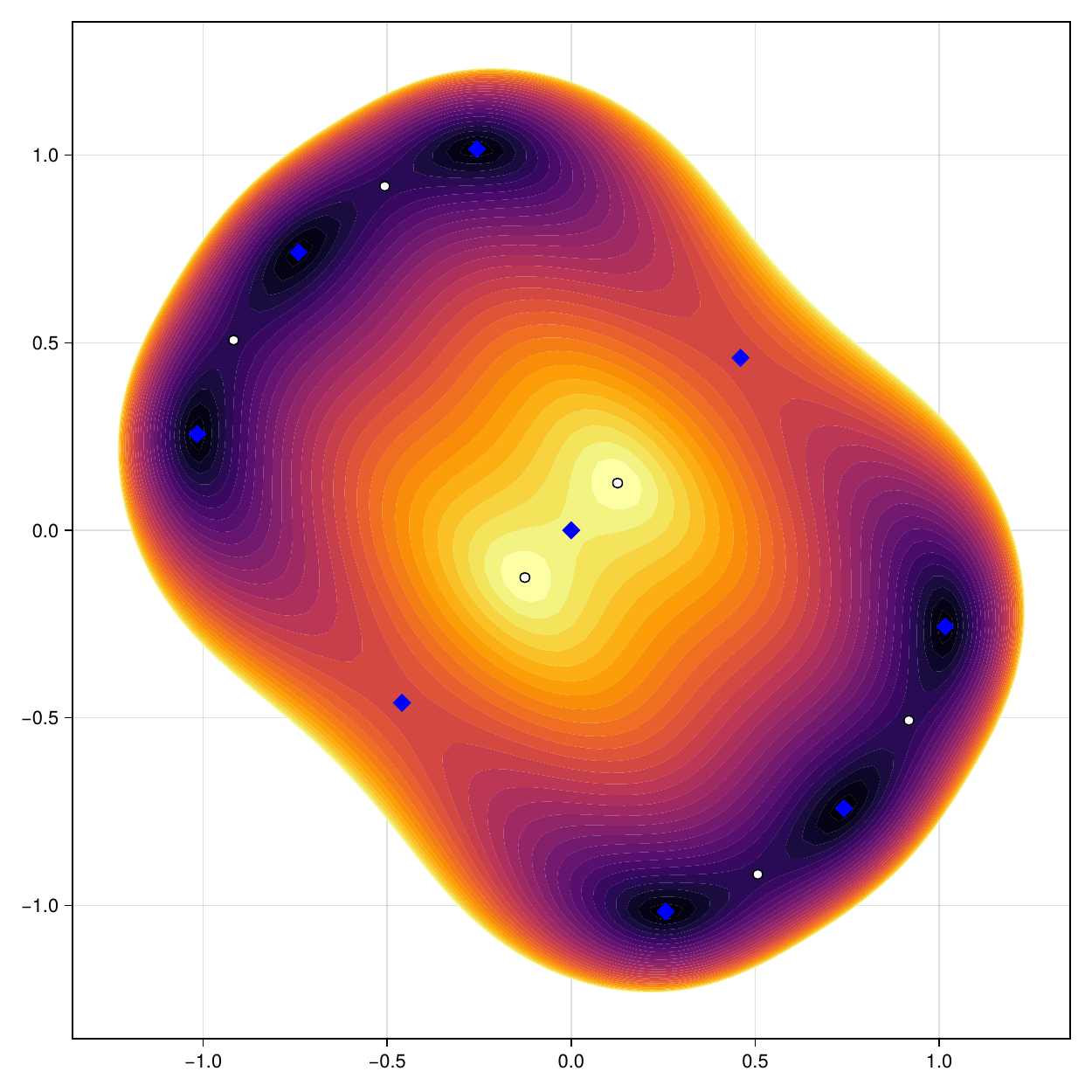}
                \caption{}
                \label{fig:subfig_b}
            \end{subfigure}           
            \caption{Critical points of $w_{22,\Ss}$ captured using Chebyshev tensorized polynomials.}
            \label{fig:deuflhard} 
    \end{figure}
    The central region contains a saddle point flanked by two local maximizers, while the overall structure includes six local minimizers, interlaced by saddle points, see \hyperref[fig:deuflhard]{Figure~\ref{fig:deuflhard}}.
    Using \href{https://github.com/gescholt/Globtim.jl}{$\mathsf{Globtim}$}, we can 
    analyze this example from two different perspectives:
    \begin{itemize}
        \item An approximate approach, where we focus on capturing the prominent features using low-degrees 
        \item An exact approach, where we aim to identify all critical points precisely and go up to high-degrees
    \end{itemize}
    At $d=8$,
    \href{https://github.com/gescholt/Globtim.jl}{$\mathsf{Globtim}$}
    identifies the six global minimizers of $f$. However, this does
    not constitute what we consider a proper capture, as the topology
    of the approximant fails to match that of the objective function.
    While we do capture the points of interest, the six local
    minimizers, we miss the two relatively small local maximizers at
    the center of the domain. The six "wells" surrounding each of the
    local minima represent less prominent features of the objective
    function. In terms of capture sequence, the central peak is most
    readily identified, followed by the two saddle points at $(-.5,
    -.5)$ and $(.5, .5)$, and finally the six local minimizers.
    Furthermore, the Legendre basis seems to slightly outperform the
    Chebyshev basis in this case in terms of capturing the local
    minimizers.
    
    \begin{figure}[!htbp]
        \centering
        \includegraphics[width=0.4\textwidth]{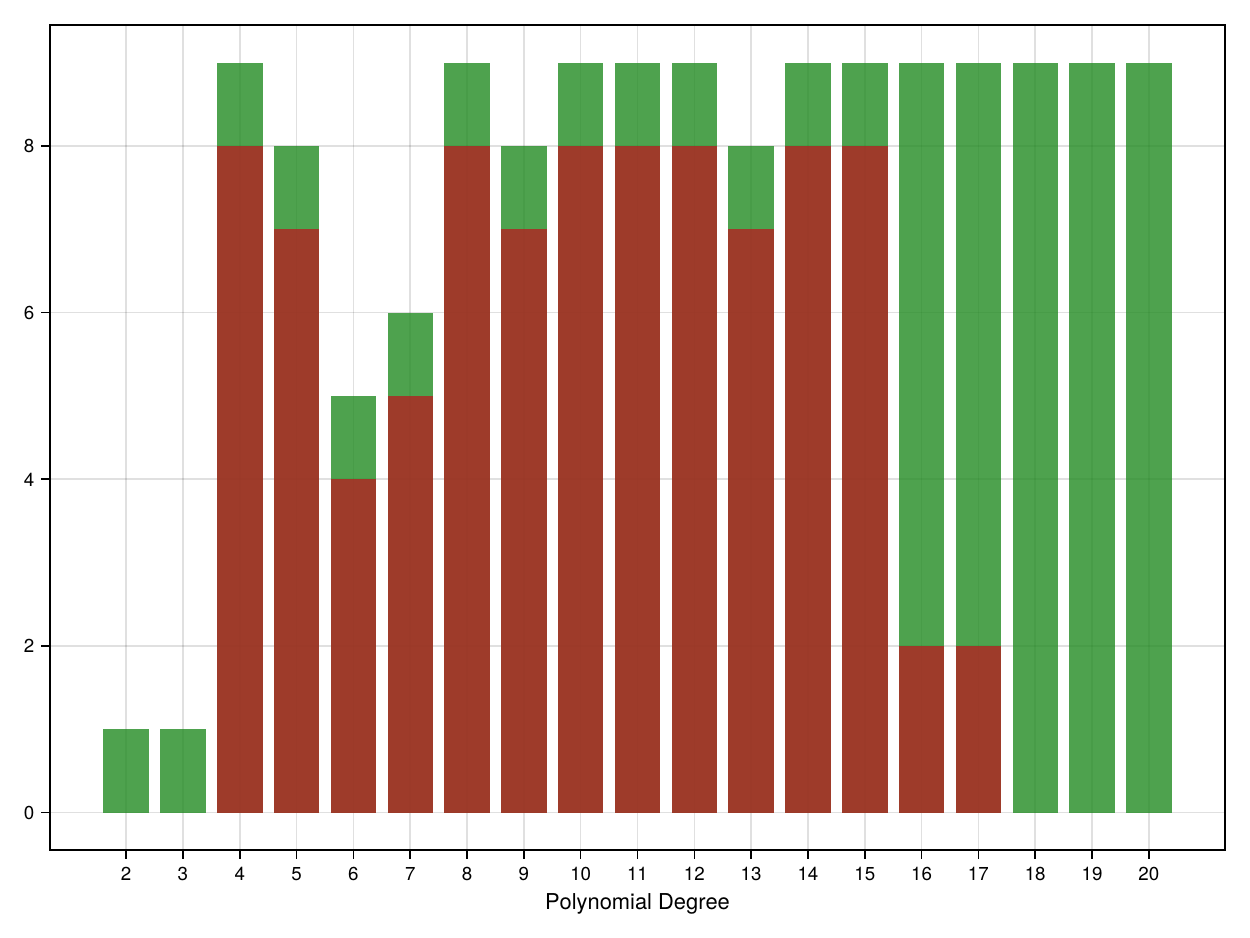}
        \caption{Total number of local minimizers globtim has computed (BFGS-optimized solutions).
        In green is the percentage of them which admit a critical point of $w_{d,\mathcal{S}}$ within a $1.e{-3}$ radius.}
        \label{fig:deufl_hist}
    \end{figure}
    Starting at degree $d=18$, both Chebyshev and Legendre approximants successfully identify all six local minima with a precision of at least $1e-3$, see Figure~\ref{fig:deufl_hist}. When we initiate local minimization on $f$ at the critical points of these approximants, the method converges to each minimizer of $f$ at least once. We also observe convergence of BFGS to saddle points in cases when the initialization point lies close enough that the function's local geometry appears effectively flat to the local optimization routine.
\end{exm}

\begin{exm}[H\"older's table]
    \label{exm:holder_table}
    We consider now the H\"older table function 2, ~\cite{Jamil2013}
    \begin{equation}
        f(x, y) = -\left\vert \sin(x) \cos(y) \exp\left(\left\vert
        1-\frac{\sqrt{x^2+y^2}}{\pi}\right\vert\right)\right\vert.
    \end{equation}
    This function has multiple local minimizers, but most importantly, it is highly non-polynomial, meaning that polynomial approximations to $f$ in the $\LL^{\infty}$-norm converge to it slowly.  
    The domain of definition we consider is $[-10, 10]^2$ which is then rescaled
    to $[-1, 1]^2$. The function attains
    its global minimum around the four corners $[8.055, 9.664]$, $[8.055,
    -9.664]$, $[-8.055, 9.664]$, $[-8.055, -9.664]$. We observe in
    Figure~\ref{fig:Holder} that the local minimizers of the
    approximant $\pol_{19}$ at degree $19$ do not match those of the objective function exactly of course,
    but the approximation is sufficiently accurate to capture all of
    $\lmin(f)$ at a reasonable accuracy.
    \begin{figure}[ht]
        \centering
        \subfloat{
            \includegraphics[width=0.45\textwidth]{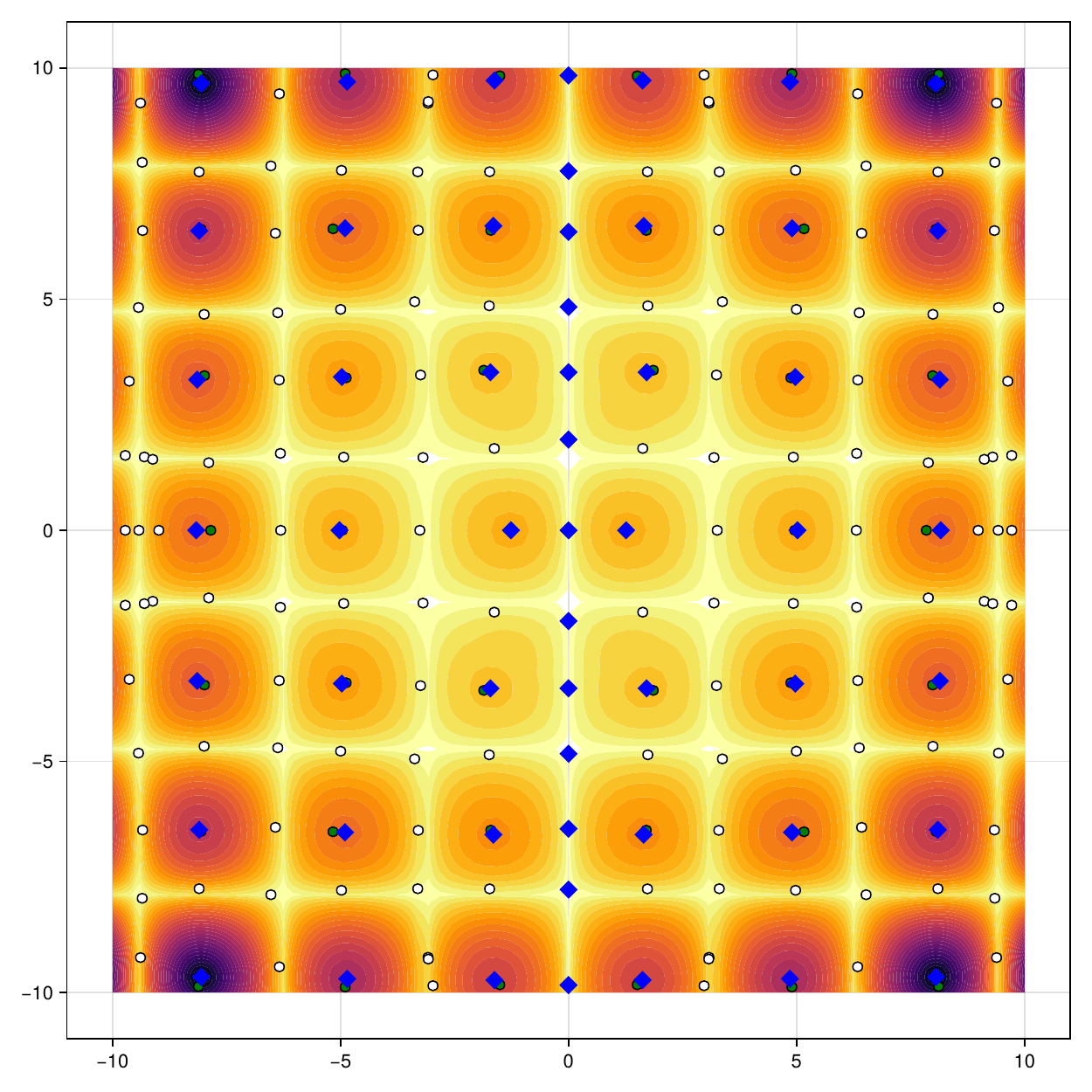}
            \label{subfig:contour_Holder}
        }
        \caption{Contour plot of the H\"older table 2 function. The blue points
        depict all the local minimizers recovered on the domain $[-10, 10]^2$ after initiating the NLPSolve Maple routine at each of the critical points of $w_{26}$, depicted in red, and computed using {\msolve}.}
        \label{fig:Holder}
    \end{figure}
    Note also that 
    {\msolve}~\cite{BES21} computes all real solutions to the polynomial system
    $\pder{w_{19}}{x}= \pder{w_{19}}{y}= 0$ in about slightly less than 10
    minutes, while running on in parallel on $10$ threads. 

    This example is interesting because of the behavior that \chebfuntwo
    exhibits on this example. It attempts to approximate $f$ with a $65537
    \times 65537$ Chebyshev coefficient matrix, which requires about 16 Gb of
    memory space to be stored.
    What \chebfuntwo is attempting to do here is to build a polynomial
    approximant of rank $\samplesize$ of the form
    \begin{equation*}
        f(x,y)\simeq \sum_{j = 1}^{k} \sigma_{j}\phi_j(y)\psi_j(x), 
    \end{equation*}
    where the ${\phi}$ and ${\psi}$ are orthonormal functions on $[-1,1]$ in the
    $\LL^2$-norm. 
    To do so, it has to first build a Chebyshev tensor grid of size
    $(2^{j+2}+1)\times (2^{j+2}+1)$ and then apply Gaussian elimination to
    reduce it to a matrix of rank $2^{j}+1$, with $j$ increasing until such a
    reduction is feasible, see~\cite[Section 2]{chebfun} for more. 
    This is then compounded with difficult computations to isolate the real
    roots of the partial derivatives of the approximant. 
    Reducing the size of the domain to $[-4, 4]^2$ was not sufficient to resolve
    this issue; computations were terminated after 24 hours of run time on a
    server, without providing satisfactory results. 
\end{exm}

\begin{exm}
\label{exm:deuflhard_4d}
\begin{figure}[htbp]
    \centering
        \begin{subfigure}[b]{0.45\textwidth}
            \centering
            \includegraphics[width=\textwidth]{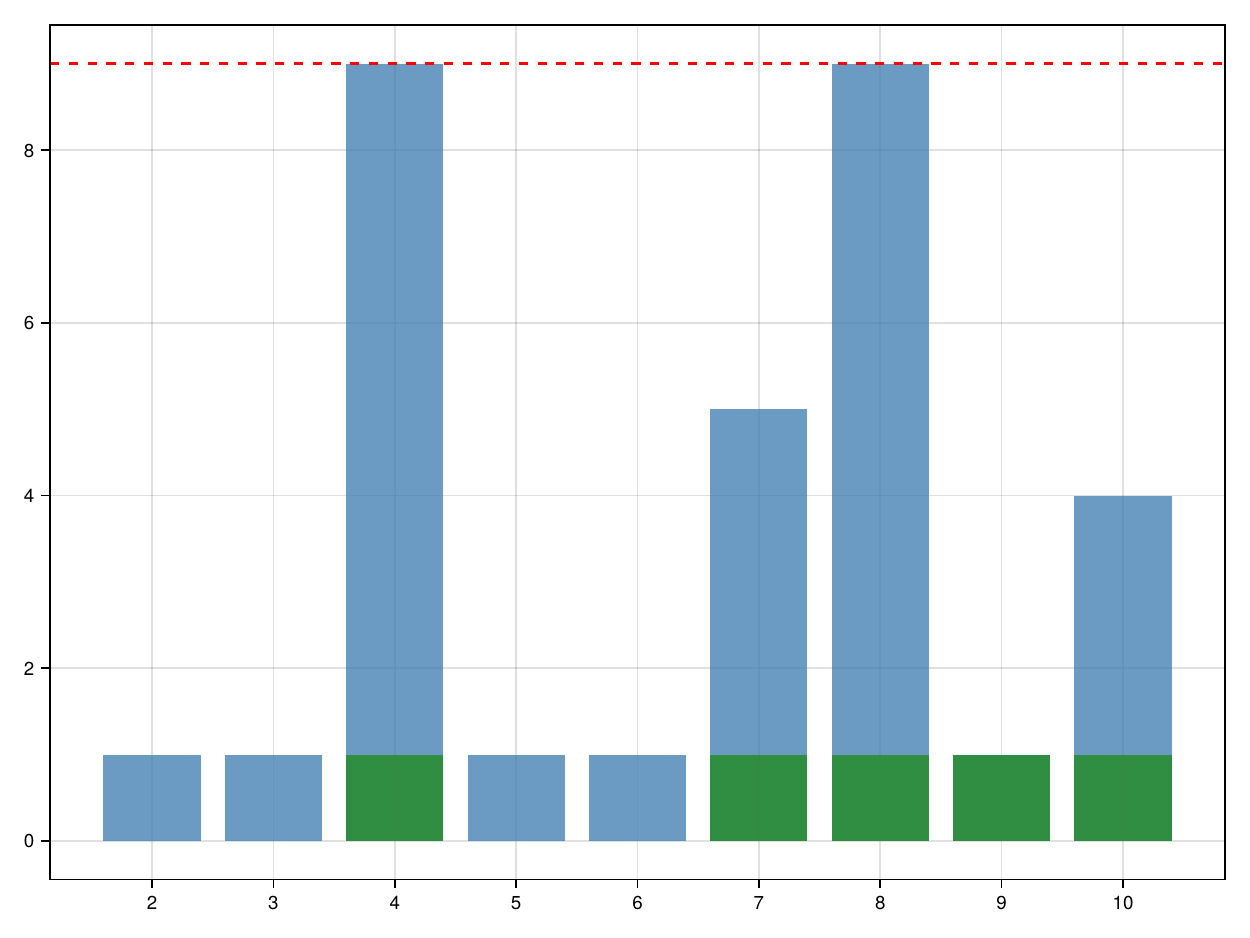}
        \end{subfigure}
        \caption{Analysis of polynomial approximation quality:  the
        captured minimizers by BFGS in blue and in green the captured
    by the un-refined critical points of $w_{d, \Ss}$.}
            \label{fig:deufl_4d_b}
\end{figure}
We now consider a more challenging problem that exceeds the capabilities of
current polynomial approximations based implementations due to the
four-dimensional optimization space. We construct a composite function by
merging two copies of the Deuflhard function: $f(x) = f(x_1,x_2) + f(x_3, x_4)$.
As demonstrated in Example~\ref{exm:deuflhard}, this function exhibits an
interesting geometry that requires polynomial approximants of high-degree
$(d\geq 18)$ for accurate capture of all critical points.  The critical points
of the composite function are obtained by tensoring the critical points of the
individual Deuflhard functions.  In order to avoid working with very high four
dimensional polynomial approximants, we restrict our analysis to a stretched
$(+,-,+,-)$ orthant: $[-0.1,1.1] \times [-1.1,0.1] \times [-0.1,1.1] \times
[-1.1,0.1]$, which contains $25$ critical points, $9$ of which are local minimizers. 
Following a methodology analogous to the 2D case, we systematically increase the
polynomial degree and collect critical points across all subdomains. But in this
case, although we are able to capture all local 9 minimizers with the refined
BFGS points for some of the degrees, see Figure~\ref{fig:deufl_4d_b}, we do not
observe a sufficient improvement in the accuracy on the location of the critical
points, even as the $\LL^2$-norm error of the approximant $w_{d, \Ss}$
decreases, see Figure~\ref{fig:deufl_analysis}.

    \begin{figure}[!htbp]
        \centering
        \begin{subfigure}[b]{0.45\textwidth}
            \centering
            \includegraphics[width=\textwidth]{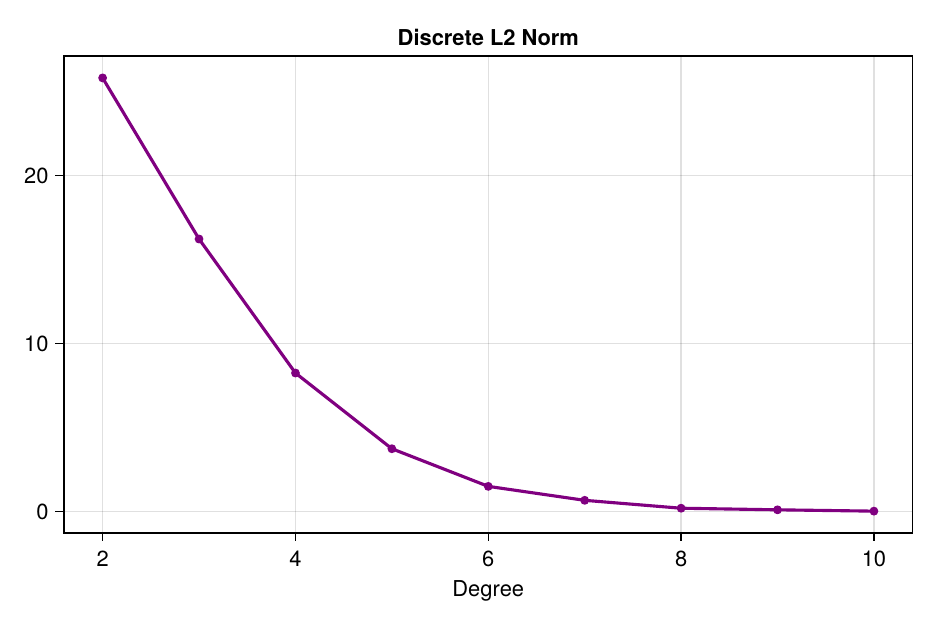}
            \caption{}
        \end{subfigure}
        \hfill
        \begin{subfigure}[b]{0.45\textwidth}
            \centering
            \includegraphics[width=\textwidth]{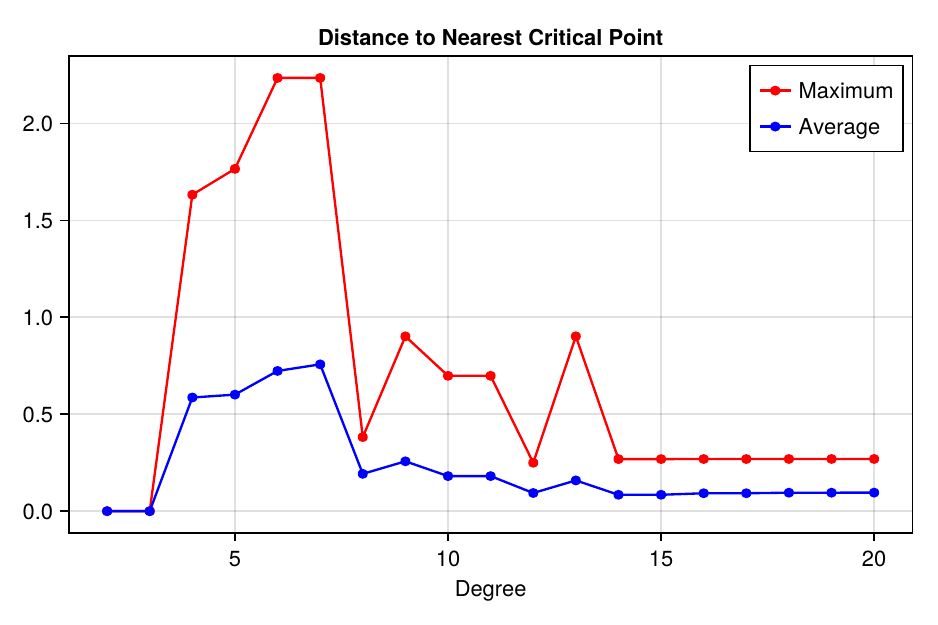}
            \caption{}
        \end{subfigure}
        \caption{Analysis of polynomial approximation quality: (a) discrete $\mathcal{L}^2$-norm error of the approximant $w_{d,\mathcal{S}}$; (b) average and maximum distances between critical points of the approximant and their BFGS-optimized values.}
        \label{fig:deufl_analysis}
    \end{figure}
At our current implementation stage, increasing the degree further could solve
this problem, but we believe it is not the best approach.  While the polynomial
expansion of our composite objective function should theoretically be
sparse---lacking cross-terms like $x_1x_3, x_1x_4, x_2x_3, x_2x_4$ due to the
additive structure---our numerical construction procedure produces approximants
with full support. 

Hence we split each dimension component into two, resulting in a total of 16 subdomains to construct approximants over.
We summarize the results of our experiments in Table~\ref{tab:deuflhard_4d_combined}.
\begin{table}[htbp]
    \centering
    \begin{tabular}{c@{\hspace{1cm}}cc@{\hspace{1.5cm}}cc}
        \toprule
        \textbf{Degree} & \multicolumn{2}{c@{\hspace{1.5cm}}}{\textbf{Minima}} & \multicolumn{2}{c}{\textbf{Saddle Points}} \\
        & \textbf{Captured} & \textbf{Max Error} & \textbf{Captured} & \textbf{Max Error} \\
        \midrule
        3 & 4/9 & 4.2e-02 & 0/16 & - \\
        4 & 4/9 & 2.6e-02 & 0/16 & - \\
        5 & 9/9 & 1.5e-01 & 16/16 & 7.7e-02 \\
        6 & 9/9 & 3.3e-01 & 16/16 & 1.6e-01 \\
        7 & 9/9 & 4.3e-01 & 16/16 & 2.1e-01 \\
        8 & 9/9 & 4.2e-01 & 16/16 & 2.1e-01 \\
        \bottomrule
    \end{tabular}
    \caption{The combined outputs of approximating the 4D Deuflhard composite function over the 16 subdomains of {$[-1.1, 1.1]^4$}. The threshold for capture of a critical point of {$f$} is set to {$0.1$}.}
    \label{tab:deuflhard_4d_combined}
\end{table}

The subdivision strategy successfully enables capture of all 9 local minimizers, even with modest polynomial degrees, and initiating BFGS at these candidate points leads to the capture of all local minimizers of $f$ at a precision of about $1e-8$. However, the table reveals an interesting pattern: the accuracy of computed critical points plateaus after $d=5$, despite continued improvement in the $\LL^2$-norm approximation error as the degree increases. 
\begin{figure}
    \centering
    \includegraphics[width=0.6\textwidth]{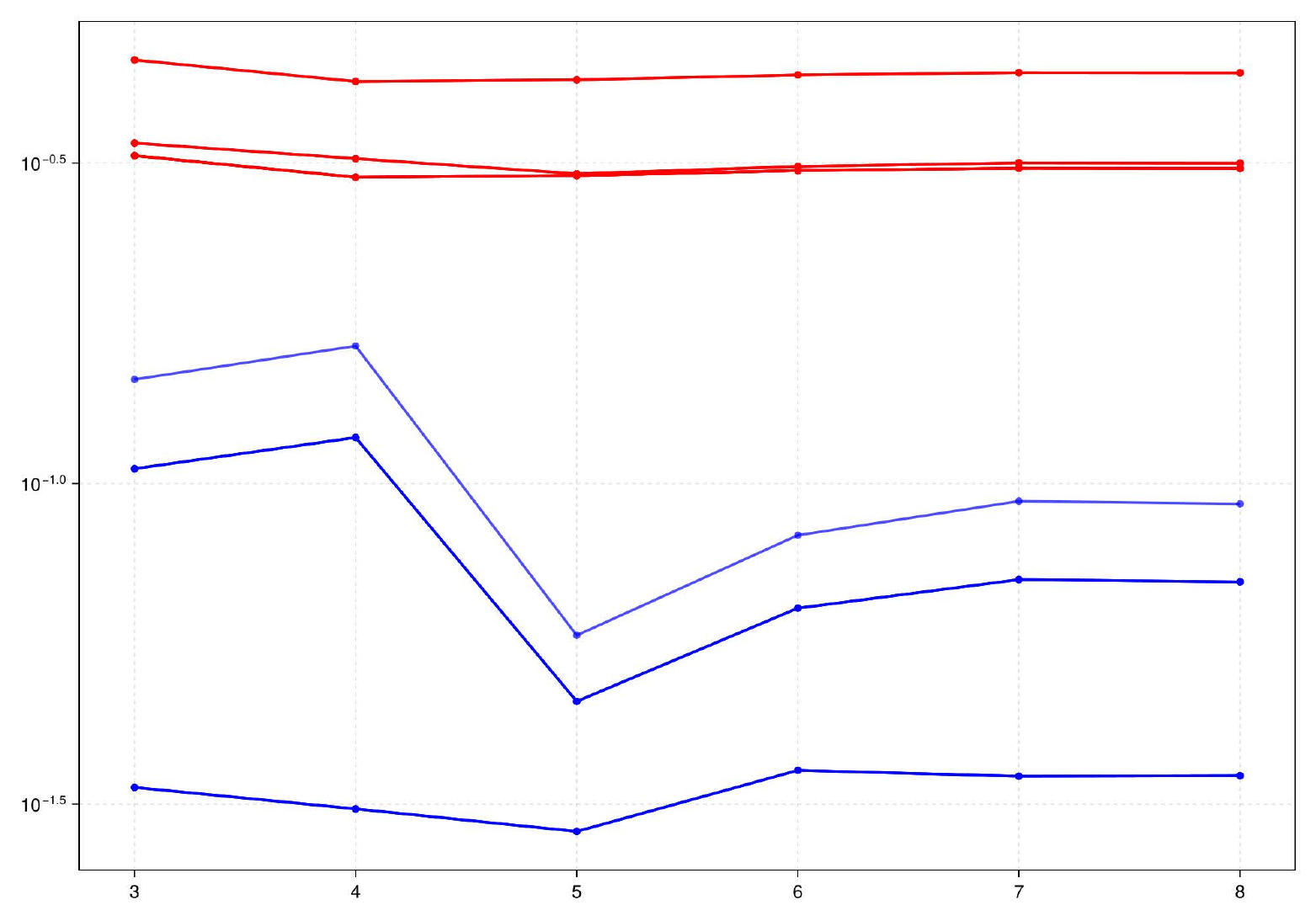}
    \caption{Combined output from all 16 subdomains. In blue the minimal distance from the each of the 9 local minimizers of $f$ to the critical points of $w_{d, \Ss}$, and in orange the maximal distance from the saddle points of $f$ to the critical points of $w_{d, \Ss}$.}
    \label{fig:deufl_4d_critical_points}
\end{figure}

This stagnation we observe in Figure~\ref{fig:deufl_4d_critical_points} stems from the disparate scales of $f$ across the domain. Near the local minimizers, $f$ attains values close to zero with relatively flat well bottoms, while at the domain boundaries (e.g., at $[1.1,-1.1,1.1,-1.1]$), $|f|$ reaches approximately 140. As we increase the polynomial degree, the least squares approximant $w_d$ naturally prioritizes accuracy where $f$ is large in norm, since improvements in these regions contribute much more to reducing the $\LL^2$-norm approximation error. The subtle geometry near the local minimizers, where $f$ is small, receives proportionally less attention in the approximation. This is corroborated by the observation that computed saddle points achieve better function value accuracy than the computed points near local minimizers, as $f$ attains larger values at saddle points.
Consequently, while our polynomial approximants successfully identify the approximate locations of local minimizers by degree 5, further degree increases yield diminishing returns for pinpointing these critical points. At this stage, transitioning to local optimization methods would provide a more efficient path to converge to the true minimizers.

\end{exm}

\bibliographystyle{siam}
\bibliography{scibib}

\end{document}